\documentclass[a4paper,pdfa,cleveref, autoref, thm-restate]{lipics-v2021}

\pdfoutput=1 
\hideLIPIcs  

\title{Computational Model for Parsing Expression Grammars} 

\author{Alexander {Rubtsov}}{HSE University, Russia \and MIPT, Russia 
}{rubtsov99@gmail.com}{https://orcid.org/0000-0001-8850-9749}{  This paper was prepared within the framework of the HSE University Basic Research Program. Supported in part by RFBR grant 20–01–00645.}

\author{Nikita {Chudinov}}{Google, Switzerland }{imonory@yandex.ru}{}{}

 \authorrunning{A. Rubtsov and N. Chudinov} 

\Copyright{Alexander Rubtsov and Nikita Chudinov}


\category{} 

\relatedversion{This paper is a full version of the conference paper}
\relatedversiondetails[cite=rubtsov_et_al:LIPIcs.MFCS.2024.80]{MFCS Proceedings}{https://doi.org/10.4230/LIPIcs.MFCS.2024.80}

\acknowledgements{Authors thank Mikhail Iumanov for helpful discussions and anonymous referees for the helpful comments.}

\nolinenumbers 

\usepackage{amsmath, amsfonts, amsthm, amssymb, amsbsy, amstext, amscd, amsxtra, multicol}
\usepackage{tikz}

\usepackage{enumitem}

\usetikzlibrary{automata, positioning, trees}
\usetikzlibrary{shapes.geometric} 
\usetikzlibrary{chains, fit, shapes, arrows, automata, positioning, snakes, arrows, arrows.meta, decorations.markings, bending}

\usepackage{tabu}

\usepackage{xcolor}

\def\REG{\mathsf{REG}}
\def\Bool{\mathsf{Bool}}
\def\DCFL{\mathsf{DCFL}}

\def\PEL{\mathsf{PEL}}

\usepackage{mathrsfs}

\def\CL{{\mathscr L}}

\def\A{{\cal A}}

\let\es\varnothing

\let\epsilon\varepsilon
\let\eps\varepsilon
\def\rendmarker{{\vartriangleleft}}
\def\lendmarker{{\vartriangleright}}

    \usepackage{mathtools}  
	\def\yield{\xRightarrow{\phantom{\scriptstyle *\; }}}
	\def\yields{\xRightarrow{{\scriptstyle *\; }}}

	\def\ders{ \vdash^{\!\!\! {}^{ *}}}    
	\def\der{ \vdash }

    \let\mder\models
    \def\mders{ \models^{\!\!\! {}^{ *}}}

	\let\lto\leftarrow
	\let\rto\rightarrow
	\let\uto\uparrow
	\let\dto\downarrow
    \let\udto\updownarrow
    
    \usepackage{graphicx}

	\let\hlto\hookleftarrow
	\let\hrto\hookrightarrow
	\let\hdto\hookdownarrow

\def\por{\mathop{/}}
\def\pnot{\text{\textbf{!}}}
\def\pfail{\text{\textbf{F}}}

\begin{document}

\maketitle

\begin{abstract}
We present a computational model for Parsing Expression Grammars (PEGs). The predecessor of PEGs top-down parsing languages (TDPLs) were discovered by A. Birman and J. Ullman in the 1960-s, B. Ford showed in 2004 that both formalisms recognize the same class named Parsing Expression Languages (PELs). A. Birman and J. Ullman established such important properties like TDPLs generate any DCFL and some non-context-free languages like $a^nb^nc^n$, a linear-time parsing algorithm was constructed as well. But since this parsing algorithm was impractical in the 60-s TDPLs were abandoned and then upgraded by B.~Ford to PEGs, so the parsing algorithm was improved (from the practical point of view) as well. Now PEGs are actively used in compilers (eg., Python replaced LL(1)-parser with a PEG one) so as for text processing as well. In this paper, we present a computational model for PEG, obtain structural properties of PELs, namely proof that PELs contain Boolean closure of regular closure of DCFLs and PELs are closed over left concatenation with regular closure of DCFLs. We present an extension of the PELs class based on the extension of our computational model. Our model is an upgrade of deterministic pushdown automata (DPDA) such that during the pop of a symbol it is allowed to return the head to the position of the push of the symbol. We provide a linear-time simulation algorithm for the 2-way version of this model, which is similar to the famous S. Cook linear-time simulation algorithm of 2-way DPDA.
\end{abstract}

\section{Introduction}
We present a computational model for Parsing Expression Grammars (PEGs) presented by B. Ford in~\cite{FordPEG2002}. The predecessor of PEGs top-down parsing languages (TDPLs) was discovered by A. Birman and J. Ullman in the 1960s (so as generalized TDPLs)~\cite{TDPL_BirmanUlman}. While the PEGs formalism has more operations, it has the same power as TDPLs and generalized TDPLs which was shown by B. Ford in~\cite{FordPEG2004}. We refer to the class of languages generated by PEGs (and TDPLs, and generalized TDPLs) as \emph{Parsing Expression Languages} (PELs).

Little is known about the structural properties of PELs. From the 60's it is known that PELs contain DCFLs as a subclass and some non-context-free languages like $a^nb^nc^n$ as well. A linear-time parsing algorithm (in RAM) had been constructed for TDPLs, but it was impractical in the 1960s since it required too much memory for memoization and TDPLs had been abandoned. 
B.~Ford upgraded the TDPLs formalism to PEGs and presented a linear-time practical algorithm in 2002~\cite{FordPEG2002}. From the Ford's construction of PEGs follows that PELs are closed over Boolean operations. It was shown in~\cite{Loff_PEG} that PEGs are somehow ``universal'', behave  as special kind of real-time Turing machine and no pumping lemma can exist for PELs. It was also shown in~\cite{Loff_PEG} that PELs contain such nontrivial languages as languages of powers~$P_k = \{a^{k^n} : n \geq 0\}$ for $k \geq 1$ and the language of palindromes that have length $2^{n+1}, n\geq 0$.
  Now PEGs are being actively used in compilers (eg., Python replaced an LL(1)-parser with a PEG one~\cite{PEP617}) so as for text processing as well. In this paper, we present a computational model for PELs and obtain some interesting properties for this class, analyze (some of) its subclasses, and generalize the PELs class as well.

A computational model for PELs was presented in~\cite{Loff_PEG}, but this model significantly differs from classical models of computations, so it is hard to clarify the place of PELs among known classes of formal languages, based on this model. So we present a simpler and more convenient model that discovers the place of PEGs in the variety of formal language classes. Namely, the computational model is a modified deterministic pushdown automaton (DPDA) that puts to the stack a symbol with the pointer of the head's position on the tape (from which the push has been performed). During the pop, the automaton has two options: either leave the head in the current position or move the head to the position stored in the pointer (retrieved during the pop of the symbol).
We call this model a deterministic pointer pushdown automaton (DPPDA).
This description of PELs from the automata point of view helped us to obtain other important results not only for the PELs but for the general area of formal languages as well. Namely, we prove that boolean closure of regular closure of DCFLs is linear-time recognizable (in RAM), what extends the nontrivial result by E.~Bertsch and M.-J.~Nederhof~\cite{Bertsch1999RegularCO} that regular closure of DCFLs is linear-time recognizable.

To describe our results we shall mention the following important results in the area of formal languages and automata theory. Donald Knuth invented LR($k$) grammars that describe DCFLs for $k \geq 1$ and were widely used in practice. It is easier to design an LL($k$) grammar for practical purposes, so despite the power of LR, LL grammars are widely used for parsing (and some artificial modification of recursive descent parsing as well). Top-down parsing languages (TDPLs, predecessor of PEGs) cover LL(1) grammars and even contain DCFLs as a subclass, but their linear-time parsing algorithm was impractical in the 1970s, so TDPLs had been abandoned till B. Ford upgraded them to PEGs and presented a practically reasonable linear-time parser (Packrat). So, linear-time recognizable classes of formal languages are used in compilers, and LR (DCFLs) parsers now compete with PEGs which cover a wider class of formal languages that is almost undiscovered. There are no comprehensive results on the structure of PELs, so we make a contribution to this open question. Another wide linear-time recognizable class of formal languages is languages recognizable by two-way deterministic pushdown automata (2DPDA). S.~Cook obtained in~\cite{Cook1970_2DPDA} a famous linear-time simulation algorithm for this model. There also was an amazing story about how D. Knuth used S.~Cook's algorithm to discover the Knuth-Morris-Pratt algorithm (\cite{KMP}, Section~7).

We modify 2DPDA in the same way as we did for DPDA: we add symbols to stack with a pointer that allows returning the head to the cell from which the push had been performed. S.~Cook's linear-time simulation algorithm applies to this model as well (with a little modification). So we extend the important class of formal languages (recognizable by 2DPDAs) preserving linear-time parsing. This extension can be used to generalize PEGs. Also, this algorithm provides another approach to linear time recognition of languages generated by PEGs described via DPPDAs. Note that there are not many structural results about PELs. Moreover, even equivalence of TDPLs and generalized TDPLs (with PEGs) had been proved by B.~Ford~\cite{FordPEG2004}  decades after these classes had been invented. In our opinion, one of the reasons for that is that TDPL-based formalisms are hard. So even the proof of inclusion DCFLs in PELs~\cite{TDPL_BirmanUlman} is complicated, while it directly follows from the equivalence of PEGs with our model.

So we hope that our model will raise interest in investigations of PELs and will help with these investigations as well. Our results also clarify the place of another interesting result (we also improved it, as described below). It was shown by E.~Bertsch and M.-J.~Nederhof~\cite{Bertsch1999RegularCO} that regular closure of DCFLs is linear-time recognizable. We show that this class is recognizable by DPPDAs which simplifies the original proof~\cite{Bertsch1999RegularCO} and shows the place of this class in the formal languages classes.

There are many linear-time recognizable classes of formal languages. Recently Rubtsov showed~\cite{RLinDLA} that Hibbard's hierarchy (the subclass of CFLs) is linear-time recognizable. So there are many open questions related to the systematization of linear time recognizable classes of formal languages and particularly the relation of Hibbard's hierarchy with languages recognizable by 1-2 DPPDAs.

\subsection{Results}

In this paper, we present a new computational model DPPDA which is equivalent to PEGs. We also consider the two-way model 2DPPDA and provide a linear time simulation algorithm for this model following S.~Cook's construction. Via DPPDA we show that the PEGs class is closed over left concatenation with regular closure of DCFLs, so PELs contain the regular closure of DCFLs as a subclass. With the linear-time simulation algorithm for 2DPPDA, we obtain another linear-time recognition algorithm for the regular closure of DCFLs and since PELs are closed over Boolean operation we prove that the Boolean closure of regular closure of DCFLs is linear-time recognizable. Note that the last result not only generalizes well known result of linear-time recognizability of regular closure of DCFLs~\cite{Bertsch1999RegularCO}, but also our proof is significantly simpler as well.




\subsection{Basic Notation}

We follow the notation from~\cite{Hopcroft1979} on formal languages, especially on context-free grammars (CFGs) and pushdown automata. We denote the input alphabet as $\Sigma$ and its elements (letters, terminals) are denoted by small letters $a, b, c,\ldots$, while letters $w, x, y, z$ denote words. The empty word is denoted by $\eps$.
 We denote nonterminals $N$ by capital letters $A, B, C,\ldots$, and $X, Y, Z$ can be used for both nonterminals and terminals.
The axiom is denoted by $S \in N$.
Words over the alphabet $N \cup \Sigma$ are called sentential forms and are denoted by small Greek letters.

\subsection{Informal Description of PEGs}

The formal definition of PEGs is not well intuitive, so we begin with an informal one that clarifies a simple idea behind this formal model. The intuition behind PEGs lies in recursive descent parsing.

One of the parsing methods for CF-grammars is a recursive descent parsing that is a process when the derivation tree is built top-down (starting from the axiom $S$) and then each nonterminal is substituted according to the associated function. A rollback is possible as well, where by rollback we mean the replacement of one production rule by another or even the replacement of the rule higher above the current node with the deletion of subtrees. This method is very general and we do not go deep into details. For our needs, we describe a recursive descent parsing of LL(1) grammars and its modification that defines PEGs.

For LL(1) grammar, the following assertion holds. Fix a leftmost derivation of a word $w\rendmarker = u a v\rendmarker$ and let $uA\alpha\rendmarker$ be a derivation step (here $\rendmarker$ is a right end marker of the input). The next leftmost derivation step is determined by the nonterminal $A$ and the terminal $a$, so the rule is the function $R(A, a)$. So, the recursive descent algorithm for an LL(1)-parser is as follows. An input $w\rendmarker$ is written in the one-way read-only tape called the \emph{input tape}. The pointer in the (constructing) derivation tree points to the leftmost nonterminal node (without children), initially the axiom $S$. This node is replaced according to the function~$R$. In the fixed above derivation step $uA\alpha\rendmarker$ the pointer is over the nonterminal $A$, $R(A, a) = xB\beta$, where $A \to xB\beta$ is a grammar rule. So, $xB\beta$ is glued into $A$ as a subtree, $x$ is a prefix of $av$ and the head of the input tape moves while scanning~$x$. If $R(A, a)$ does not contain a nonterminal, then (after replacement) the tree is traversed via DFS until the next (leftmost!) nonterminal is met. Each terminal during this traversal shifts the head of the input tape. If the symbol under the head differs from the traversed terminal, the input word is rejected. We illustrated the described process in Fig.~\ref{Fig:LL1Traversal}. Note that $u$, $x$, $\alpha$, $\beta$ are the subtrees and $u$, $v$, $x$, $v'$ in fact occupies several cells of the input tape.

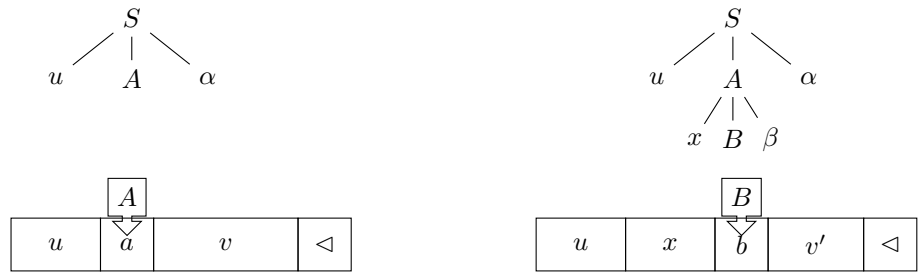
\begin{figure}[ht]
\begin{center}    
	\begin{multicols}{2}
		\begin{tikzpicture}  
            [triangle/.style = {regular polygon, regular polygon sides=3, draw},
            level 1/.style={sibling distance=10mm, level distance=8mm},            
             level 2/.style={sibling distance=5mm, level distance=10mm}]
			\node   (S)   {$S$}  
            child  {node {$u$}}
            child  {node {$A$}                   
            }			                      			
			child {
					node {$\alpha$}
			};
            
            	\edef\sizetape{0.7cm}
            	\tikzstyle{tape}=[draw,minimum size=\sizetape]
            	\tikzstyle{head}=[arrow box,draw,minimum size=.5cm,arrow box
            	arrows={south:.25cm}] 
            	\tikzstyle{arrup}=[single arrow, draw, shape border rotate=90, fill=gray!30, single arrow head extend=.001cm, single arrow head indent=.01ex]
            	\tikzstyle{stack}=[rectangle split, rectangle split parts=#1,draw, anchor=center, minimum width = 1.5cm]
            	\tikzstyle{stack_cell}=[rectangle, draw, anchor=center, minimum width = 1.5cm ]


            	\begin{scope}[start chain=tape going right,node distance=-0.15mm]
            	    \node [on chain=tape,tape] at (-1,-3) (input) {\phantom{xx}$u$\phantom{xx}};
                        	    \node [on chain=tape,tape] {$a$};
                        	    \node [on chain=tape,tape] {\phantom{xxxx}$v$\phantom{xxxx}};
            		\node [on chain=tape,tape] {$\rendmarker$};		
		
            	\end{scope}
            	\newcommand{\head}[2]{\node [head, yshift=.27cm] at (tape-#1.north) (head#1) {#2};
            	\node [yshift=-.5cm] at (head#1.south)  (head#1_cpoint) {};	
            	}
                \head{2}{$A$}						
        \end{tikzpicture}
		\begin{tikzpicture}
            [level 1/.style={sibling distance=10mm, level distance=8mm},            
             level 2/.style={sibling distance=5mm, level distance=8mm}]
			\node   (S)   {$S$}  
            child  {node {$u$}}
            child  {node {$A$}
                    child{ node {$x$} }
                    child{ node {$B$} }
                    child{ node {$\beta$} }
            }			                      			
			child {
					node {$\alpha$}
			};
                        	\edef\sizetape{0.7cm}
                        	\tikzstyle{tape}=[draw,minimum size=\sizetape]
                        	\tikzstyle{head}=[arrow box,draw,minimum size=.5cm,arrow box
                        	arrows={south:.25cm}] 
                        	\tikzstyle{arrup}=[single arrow, draw, shape border rotate=90, fill=gray!30, single arrow head extend=.001cm, single arrow head indent=.01ex]
                        	\tikzstyle{stack}=[rectangle split, rectangle split parts=#1,draw, anchor=center, minimum width = 1.5cm]
                        	\tikzstyle{stack_cell}=[rectangle, draw, anchor=center, minimum width = 1.5cm ]


                        	\begin{scope}[start chain=tape going right,node distance=-0.15mm]
                        	    \node [on chain=tape,tape] at (-2,-3) (input) {\phantom{xx}$u$\phantom{xx}};
                        	    \node [on chain=tape,tape] {\phantom{xx}$x$\phantom{xx}};
                        	    \node [on chain=tape,tape] {$b$};
                                \node [on chain=tape,tape] {\phantom{xx}$v'$\phantom{xx}};
                        		\node [on chain=tape,tape] {$\rendmarker$};		
		
                        	\end{scope}
                        	\newcommand{\head}[2]{\node [head, yshift=.27cm] at (tape-#1.north) (head#1) {#2};
                        	\node [yshift=-.5cm] at (head#1.south)  (head#1_cpoint) {};	
                        	}
                            \head{3}{$B$}									
        \end{tikzpicture}
	\end{multicols}
\end{center}    
		\caption{Example of LL(1) recursive descent parsing}\label{Fig:LL1Traversal}
\end{figure}

So now we move to the description of PEGs via modification of recursive descent parsing. 
In the first example, we will provide similar PEG and CFG (Fig.~\ref{Fig:PEGvsCFG}) and explain their similarity and differences.

\begin{figure}[h]
    \begin{center}            
        \begin{minipage}{0.5\textwidth}          
		    \begin{multicols}{2}
             \begin{minipage}{0.5\textwidth}
                \hspace*{3em}\textbf{PEG}	
		        \vspace*{-5pt}		
		        \begin{align*}
		        	S &\lto AB \por BC \\
		        	A &\lto aA \por a\\
		        	B &\lto abb \por b \\ 
                    C &\lto cC \por \eps
		        \end{align*}
              \end{minipage}
			\columnbreak
            \begin{minipage}{0.5\textwidth}
                \hspace*{3em}\textbf{CFG}
                \vspace*{-5pt}			            
			    \begin{align*}
			    	S &\to AB \mid BC \\
			    	A &\to aA \mid a\\
			    	B &\to abb \mid b \\ 
                    C &\to cC \mid \eps 
			    \end{align*}	
         \end{minipage}	
		\end{multicols}
      \end{minipage}
    \end{center}
        \caption{PEG and CFG for comparison}\label{Fig:PEGvsCFG}
\end{figure}

PEGs look similar to context-free grammars, but the meaning of almost all concepts are different, therefore the arrow $\lto$ is used to separate the left part of a rule from the right part. The difference comes from the following approach to recursive descent parsing. 
We describe the PEG via the transformation of the CFG. Let us order all the rules of the CFG for each nonterminal. During recursive descent parsing, we will try each rule according to this order. 
If a failure happens, let us try the next rule in the order. If the last rule leads us to the failure too, propagate the failure to the parent and try using the next rule in the order on the previous tree level. So, that is the reason why all right-hand sides of the rules in PEG are separated by the delimiter $\por$, but not by $\mid$. The order of rules in PEGs matters, unlike CFGs. Consider the parsing (Fig.~\ref{Fig:PEGparsingEx}) of the word $aab$ by the PEG defined on Fig.~\ref{Fig:PEGvsCFG}.

	\begin{figure}

	\begin{multicols}{4}
		\begin{tikzpicture}
           [level 1/.style={sibling distance=10mm},            
            level 2/.style={sibling distance=5mm, level distance=10mm}]
			\node   (S)   {$S$}  
			child  {
				node  (A) {$A$}
				child{
					node	{$a$}
                }
                child{
                    node {$A$}
                    child{node {$a$}}
                    child{node {$A$}
                        child{node {$\underline a$}}
                        child{node {$A$}}
                    }        
				}				
			}
			child {
					node (B) {$B$}
			};
            	\edef\sizetape{0.5cm}
            	\tikzstyle{tape}=[draw,minimum size=\sizetape]
            	\tikzstyle{head}=[arrow box,draw,minimum size=.35cm,arrow box
            	arrows={south:.2cm}] 
            	\tikzstyle{arrup}=[single arrow, draw, shape border rotate=90, fill=gray!30, single arrow head extend=.001cm, single arrow head indent=.01ex]
            	\tikzstyle{stack}=[rectangle split, rectangle split parts=#1,draw, anchor=center, minimum width = 1.5cm]
            	\tikzstyle{stack_cell}=[rectangle, draw, anchor=center, minimum width = 1.5cm ]


            	\begin{scope}[start chain=tape going right,node distance=-0.15mm]
            	    \node [on chain=tape,tape] at (-0.75,-6) (input) {$a$};
                        	    \node [on chain=tape,tape] {$a$};
                        	    \node [on chain=tape,tape] {$b$};
            		\node [on chain=tape,tape] {$\rendmarker$};		
		
            	\end{scope}
            	\newcommand{\head}[2]{\node [head, yshift=.27cm] at (tape-#1.north) (head#1) {#2};
            	\node [yshift=-.5cm] at (head#1.south)  (head#1_cpoint) {};	
            	}
                \head{3}{$a$}			
        \end{tikzpicture}
		\begin{tikzpicture}
           [level 1/.style={sibling distance=10mm},            
            level 2/.style={sibling distance=5mm, level distance=10mm}]
			\node   (S)   {$S$}  
			child  {
				node  (A) {$A$}
				child{
					node	{$a$}
                }
                child{
                    node {$A$}
                    child{node {$a$}}
                    child{node {$A$}
                        child{node {$\underline a$}}                        
                    }        
				}				
			}
			child {
					node (B) {$B$}
			};
            	\edef\sizetape{0.5cm}
            	\tikzstyle{tape}=[draw,minimum size=\sizetape]
            	\tikzstyle{head}=[arrow box,draw,minimum size=.35cm,arrow box
            	arrows={south:.2cm}] 
            	
            	\tikzstyle{arrup}=[single arrow, draw, shape border rotate=90, fill=gray!30, single arrow head extend=.001cm, single arrow head indent=.01ex]
            	\tikzstyle{stack}=[rectangle split, rectangle split parts=#1,draw, anchor=center, minimum width = 1.5cm]
            	\tikzstyle{stack_cell}=[rectangle, draw, anchor=center, minimum width = 1.5cm ]


            	\begin{scope}[start chain=tape going right,node distance=-0.15mm]
            	    \node [on chain=tape,tape] at (-0.75,-6) (input) {$a$};
                        	    \node [on chain=tape,tape] {$a$};
                        	    \node [on chain=tape,tape] {$b$};
            		\node [on chain=tape,tape] {$\rendmarker$};		
		
            	\end{scope}
            	\newcommand{\head}[2]{\node [head, yshift=.27cm] at (tape-#1.north) (head#1) {#2};
            	\node [yshift=-.5cm] at (head#1.south)  (head#1_cpoint) {};	
            	}
                \head{3}{$a$}			
        \end{tikzpicture}
		\begin{tikzpicture}
           [level 1/.style={sibling distance=12mm},            
            level 2/.style={sibling distance=5mm, level distance=10mm}]
			\node   (S)   {$S$}  
			child  {
				node  (A) {$A$}
				child{
					node	{$a$}
                }
                child{
                    node {$A$}
                    child{node {$a$}}                    
				}				
			}
			child {
					node (B) {$B$}
                    child { node {$\underline{a}$} }
                    child { node {$b$} }
                    child { node {$b$} }
			};
            	\edef\sizetape{0.5cm}
            	\tikzstyle{tape}=[draw,minimum size=\sizetape]
            	\tikzstyle{head}=[arrow box,draw,minimum size=.35cm,arrow box
            	arrows={south:.2cm}] 
            	
            	\tikzstyle{arrup}=[single arrow, draw, shape border rotate=90, fill=gray!30, single arrow head extend=.001cm, single arrow head indent=.01ex]
            	\tikzstyle{stack}=[rectangle split, rectangle split parts=#1,draw, anchor=center, minimum width = 1.5cm]
            	\tikzstyle{stack_cell}=[rectangle, draw, anchor=center, minimum width = 1.5cm ]


            	\begin{scope}[start chain=tape going right,node distance=-0.15mm]
            	    \node [on chain=tape,tape] at (-0.75,-6) (input) {$a$};
                        	    \node [on chain=tape,tape] {$a$};
                        	    \node [on chain=tape,tape] {$b$};
            		\node [on chain=tape,tape] {$\rendmarker$};		
		
            	\end{scope}
            	\newcommand{\head}[2]{\node [head, yshift=.27cm] at (tape-#1.north) (head#1) {#2};
            	\node [yshift=-.5cm] at (head#1.south)  (head#1_cpoint) {};	
            	}
                \head{3}{$a$}			
        \end{tikzpicture}
		\begin{tikzpicture}
           [level 1/.style={sibling distance=10mm},            
            level 2/.style={sibling distance=5mm, level distance=10mm}]
			\node   (S)   {$S$}  
			child  {
				node  (A) {$A$}
				child{
					node	{$a$}
                }
                child{
                    node {$A$}
                    child{node {$a$}}                    
				}				
			}
			child {
					node (B) {$B$}
                    child { node {$\underline  b$} }
			};
            	\edef\sizetape{0.5cm}
            	\tikzstyle{tape}=[draw,minimum size=\sizetape]
            	\tikzstyle{head}=[arrow box,draw,minimum size=.35cm,arrow box
            	arrows={south:.2cm}] 
            	
            	\tikzstyle{arrup}=[single arrow, draw, shape border rotate=90, fill=gray!30, single arrow head extend=.001cm, single arrow head indent=.01ex]
            	\tikzstyle{stack}=[rectangle split, rectangle split parts=#1,draw, anchor=center, minimum width = 1.5cm]
            	\tikzstyle{stack_cell}=[rectangle, draw, anchor=center, minimum width = 1.5cm ]


            	\begin{scope}[start chain=tape going right,node distance=-0.15mm]
            	    \node [on chain=tape,tape] at (-0.75,-6) (input) {$a$};
                        	    \node [on chain=tape,tape] {$a$};
                        	    \node [on chain=tape,tape] {$b$};
            		\node [on chain=tape,tape] {$\rendmarker$};		
		
            	\end{scope}
            	\newcommand{\head}[2]{\node [head, yshift=.27cm] at (tape-#1.north) (head#1) {#2};
            	\node [yshift=-.5cm] at (head#1.south)  (head#1_cpoint) {};	
            	}
                \head{3}{$b$}			
        \end{tikzpicture}
	\end{multicols}
		\caption{Parsing of $aab$ by PEG}\label{Fig:PEGparsingEx}
	\end{figure}
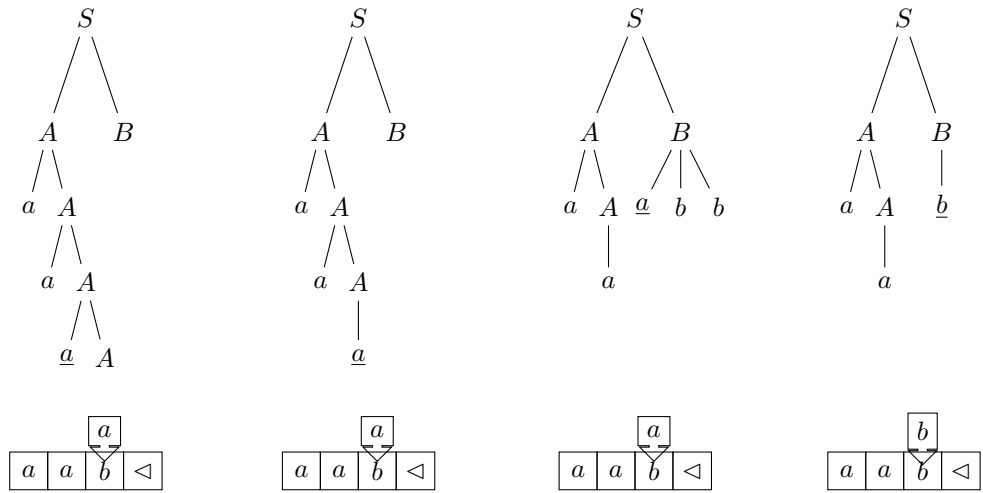
    
The rule $A \lto aA $ is applied while the content of the input tape matches the crown (the leafs) of the tree. So, when the last application is unsuccessful, it is replaced by the following rule $A \lto a$ which is unsuccessful too. So failure signal goes to the level above and the second rule $A \lto aA $ is replaced by $A \lto a $. After that, the control goes to the nonterminal $B$ for which firstly the rule $B \lto abb $ is applied, but since it leads to the failure, finally the rule $B \lto b$ is applied and it finishes the parsing since the whole word has been matched.

So PEGs are similar to CFGs since they share the idea of recursive descent parsing. But the difference is significant. Since all the rules for each nonterminal are ordered, the classical notion of concatenation does not apply to PEGs. We cannot say that if a word $u$ is derived from $A$ and $v$ is derived from $B$, then $uv$ is derived from $AB$ as explained below. In the PEG example above, a word $abb$ is never derived from $B$ because $A$ from $AB$ will always parse all $a$'s from the input. Note that the failure during the parsing occurs only because of a mismatch. So, the input $abbc$ will be parsed by the PEG as follows. 
 The prefix $ab$ will be successfully parsed by $AB$ and by $S$ as well, but since the whole word has not been parsed, the input is rejected. Since there was no failure, the rule $S \lto AB$ was not replaced by $S \lto BC$. So the word $abbc$ is not accepted by the PEG while it is derived from the CFG.
\begin{remark}\label{Remark:ConcatClosure}
    It is an open question, whether PELs are closed over concatenation.
\end{remark}

Note that the patterns of iteration $A \lto a A \por a$ and $C \lto c C \por \eps$ work in a greedy way. In the case of concatenation $Ce$ (with an expression $e$), all $c$'s from the prefix of the input would be parsed by $C$.

In the considered example we have not mentioned an important PEG's operation.
There is a unary operator $\pnot$ that is applied as follows. In the case $\pnot e$ the following happens. Firstly the parsing goes to the expression $e$.
 If $e$ parsed the following input successfully (i.e., a subtree for $e$ that matches the prefix of the unprocessed part of the input has been constructed without a failure), then $\pnot e$ produces failure. If a failure happens, then $\pnot e$ 
is considered to parse the empty word $\eps$ and the parsing process continues. For example, consider the following PEG:
\[
S \lto A (\pnot C) \por B \quad A \lto a A b \por \eps \quad B \lto a  B c \por \eps \quad C \lto a \por b.
\]
$(\pnot C)$ guarantees that if $A (\pnot C)$ finished without failure, then it parsed the whole input. 
So, in the case of the input $a^nb^n$ for $n \geq 0$, the input will be parsed by $A (\pnot C)$ and there will be no switch to the rule $S \lto B$. For any other input, the parsing of $A (\pnot C)$ fails and the rule is switched to $S \lto B$.
So, this PEG generates the language $\{a^nb^n \mid n \geq 0\} \cup \{a^nc^n \mid n \geq 0\}$.

Another common use of the operator $\pnot$ is its double application that has its name: $\& e = \pnot(\pnot e)$. This construction checks whether the prefix of the (unprocessed part of the) input matches $e$: in the positive case, $\&$ parses an empty word and computation continues, in the negative case, it returns failure. So $\&$ acts similarly to $\pnot$, but the conditions are flipped. Consider the following example:
\[
S \lto (\& (A c) ) B C \quad A \lto a A b \por \eps \quad B \lto a  B \por a \quad C \lto b C c \por \eps.
\]
This PEG checks that the input has the prefix $a^nb^nc$ and then parses the input if it has the form $a^{+}b^nc^n$, so the PEG generates the language $\{a^nb^nc^n \mid n \geq 1\}$.

So it is known that PEGs generate non CFLs and it is still an open question whether PEGs generate all CFLs. The conditional answer is no: there exists a linear-time parsing algorithm for PEG, while the work of L.~Lee~\cite{Lee02} and Abboud et al.~\cite{AbboudSIAM18} proves that it is very unlikely for CFLs due to theoretical-complexity assumptions: any CFG parser with time complexity $O(gn^{3-\eps})$, where $g$ is the size of the grammar and $n$ is the length of the input word, can be efficiently converted into an algorithm to multiply $m\times m$ Boolean matrices in time $O(m^{3-\eps/3})$. Note that this conditional result shows that it is unlikely that 2DPPDAs recognize all CFLs as well.

\section{Formal Definition of PEGs}

Our definition slightly differs from the standard definition of PEG from \cite{FordPEG2004} (Section 3) due to technical reasons. We discuss the difference after the formal definition.

\begin{definition}\label{Def:PEG}
    A parsing expression grammar $G$ is defined by a tuple $(N, \Sigma, P, S)$,
     where $N$ is a finite set of symbols called \emph{nonterminals},
     $\Sigma$ is a finite input alphabet (a set of \emph{terminals}), $N \cap \Sigma = \es$,
     $S \in N$ is \emph{the axiom}, and $P$ is a set of \emph{production rules} of the form $A \lto e$ 
     such that each nonterminal $A \in N$ has the only corresponding rule, 
     and $e$ is an expression that is defined recursively as follows. The empty word $\eps$, a terminal $a \in \Sigma$, and a nonterminal $A \in N$ are expressions.
     If $e$ and $e'$ are expressions, than so are $(e)$ which is equivalent to~$e$, a \emph{sequence} $e e'$, 
     a \emph{prioritized choice} $e \por e'$, a \emph{not predicate} $\pnot e$.
    We assume that $\pnot$ has the highest priority, the next priority has the sequence operation and the prioritized choice has the lowest one.
    We denote the set of all expressions over $G$ by $E_G$ or by $E$ if the grammar is fixed.
\end{definition}

To define the language generated by a PEG $G$ we define recursively a partial function
 $R: E \times \Sigma^* \to (\Sigma^* \cup \{\pfail\}) $
  that takes as input the expression $e$, the input word $w$, and if $R(e, w) = s \in \Sigma^*$, then $s$ is the suffix of $w = ps$ such that the prefix $p$ has been parsed by $e$ during the processing of $w$; if $R(e, w) = \pfail$ it indicates a failure that happens during the parsing process. So, the function $R$ is defined recursively as follows:
  \begin{itemize}
      \item  $R(\eps, w) = w$, $R(a, as) = s$, $R(a, bs) = \pfail$ (where $a \neq b$)
      \item  $R(e_1e_2, w) = R(e_2, R(e_1, w))$ if $R(e_1, w) \neq \pfail$, otherwise $R(e_1e_2, w) = \pfail$ 
      \item  $R(A, w) = R(e, w)$, where $A \lto e \in P$
      \item  $R(e_1 \por e_2, w) = R(e_1, w)$ if $R(e_1, w) \neq \pfail$, otherwise $R(e_1 \por e_2, w) = R(e_2, w)$
      \item  $R(\pnot e, w) = w$ if $R(e, w) = \pfail$, otherwise $R(\pnot e, w) = \pfail$
  \end{itemize}
  
Note that $R(e, w)$ is undefined if during the recursive computation, $R$ comes to an infinite loop. In fact, we will never meet this case because for each PEG there exists an equivalent form for which $R$ is a total function (see  Subsection~\ref{subsect::PEGForms}).

  We say that a PEG $G$ \emph{generates} the language $L(G) = \{ w \mid R(S, w) = \eps \}$; if $R(S, w) = \eps$ we say that $w$ is \emph{generated} by $G$. 

\subsection{Difference with other standard definitions and forms of PEGs}\label{subsect::PEGForms}

Note that our definition of $L(G)$ differs from \cite{FordPEG2004} (Section 3). The difference is about the operations allowed in PEG and the acceptance condition as well. In this subsection, we explain the difference and provide an overview of different forms of PEGs.

In the case of practical parsing, it is convenient to have more operations in the definition of PEG, but theoretically, it is more convenient to have fewer operations for the sake of the proofs' simplicity. In \cite{FordPEG2004} B.~Ford investigated different forms of PEGs and proved their equivalence, so as the equivalence with (generalized) top-down parsing languages. We begin our overview with operations that are so easy to express via operations from our definitions that they can be considered (as programmers say) syntactic sugar:
\begin{itemize}
    \item \textbf{Iterations:} $e^*$ is equivalent to $A \lto e A \por \eps$;\; $e^+ = ee^*$
    \item \textbf{Option expression:} $e?$ is equivalent to $A \lto e \por \eps$
    \item \textbf{And predicate:} $\& e = \pnot(\pnot e)$
    \item \textbf{Any character:} $\bullet  = a_1 \por a_2 \por \ldots \por a_k$ where $\Sigma = \{a_1,\ldots, a_k\}$ %
    \item \textbf{Failure:} $\pfail = \pnot \eps$ (we use the same notation as for the failure result)
\end{itemize}
We can use these constructions below. In this case, the reader can assume that they are reduced to the operations from Definition~\ref{Def:PEG} as we have described.

So by adding to the definition (or removing) syntactic sugar operations, one obviously obtains an equivalent definition (in terms of recognizable languages' class). Now we move to the nontrivial cases proved in~\cite{FordPEG2004}.

A PEG $G$ is \emph{complete} if for each $w \in \Sigma^*$ the function $R(S, w)$ is defined. A PEG $G$ is \emph{well-formed} if it does not contain directly or mutually left-recursive rules, such as $A \lto A a \por a$. It is easy to see that a well-formed grammar is complete. It was proved in~\cite{FordPEG2004} that each PEG has an equivalent well-formed one and the algorithm of the transformation had been provided as well. So from now on we assume that each PEG in our constructions is well-formed. Note that most PEGs that are used in practice are well formed by construction.

Another interesting result from~\cite{FordPEG2004} is that each PEG has an equivalent one without predicate $\pnot$. Despite this fact, we decided to include $\pnot$ in our definition since unlike substitutions for syntactical sugar operations, removing $\pnot$ predicate requires significant transformations of the PEG. Since $\pnot$ predicate is widely used in practice and it does not affect our constructions, by including $\pnot$ in the definition we achieve the constructions that can be used in practice.

As we have already mentioned our condition of the input acceptance also differs from~\cite{FordPEG2004}. We used the provided approach since if $R(S, w) = \eps$ we can reconstruct the parsing tree with the root $S$ that generates $w$. We use this property for the transformation of PEG to the computational model and the inverse transformation as well. Firstly, in~\cite{FordPEG2004} there is no axiom in PEG, but there is a starting expression $e_S$. This difference is insignificant since one can state $e_S = S$ and $S \lto e_S$ for the opposite direction. A PEG from~\cite{FordPEG2004} generates the input $w$ if $R(e_S, w) \neq \pfail$, so $R(e_S, w) = y$, where $w = xy$. So to translate PEG from~\cite{FordPEG2004} to ours one needs to set $S \lto e_S (\bullet)^*$. The transformation in the other direction is $e_S = S (\pnot \bullet)$.

\section{Definition of the Computational Model}

We call our model \emph{deterministic pointer pushdown automata} (DPPDA). We consider a one-way model (1DPPDA or just DPPDA) as a restricted case of a two-way model (2DPPDA), so we define the two-way model only.

\begin{definition}\label{def::BA}
A $2$-way deterministic pointer pushdown automata $M$
        is defined by a tuple \[\langle Q, \Sigma_{\lendmarker\rendmarker}, \Gamma, F, q_0, Z_0, \delta  \rangle, \text{ where }\]
	\begin{itemize}
		\item $Q$ is the finite set of automaton states.
		\item $\Sigma_{\lendmarker\rendmarker} = \Sigma\cup\{\lendmarker,\rendmarker\}$, where $\Sigma$ is the finite input alphabet and ${\lendmarker,\rendmarker}$ are the end markers. The input has the form $\lendmarker w\rendmarker,\, w\in \Sigma^*$.		
        \item $\Gamma$ is the alphabet of the pushdown storage.
		\item $F \subseteq Q$ is the set of the final states.
		\item $q_0 \in Q$ is the initial state.
        \item $Z_0 \in \Gamma$ is the initial symbol in the pushdown storage.
		\item $\delta$ is the partial transition function defined as
	$\delta : Q \times \Sigma_{\lendmarker\rendmarker} \times \Gamma \to
	Q\times\Gamma^*_\eps\times\{\lto, \dto, \uto, \rto \}$,
     where $\Gamma_\eps = \Gamma \cup \{\eps\}$. 
     Moreover, if $\delta(q, a, z) = (q', \alpha, \uto)$, then $\alpha = \eps$.
	\end{itemize}
\end{definition}

To define and operate with configurations of automata we introduce some notation. We denoted by $\alpha \times \vec{i} = (Z_m, i_m), \ldots, (Z_0, i_0)$ the zip of the sequences $\alpha$ and $\vec i$, which are of the same length by the definition. A right associative operation $x:\vec{l}$ prepends an element $x$ to the beginning of the vector $\vec{l}$ (we adopt this operator from Haskell programming language). E.g., if $\vec{l} = 1,2,3$ and $\vec{r} = 2,3$, we write $\vec{l} = 1:\vec{r}$.

A configuration of $M$ on a word $w$ is a quadruple $c\in Q\times (\Gamma \times I)^* \times I$, where $I = \{0,\ldots, |w|+1\}$; we refer to $w_i, i \in I$ as the $i$-th input symbol; $w_0 = \lendmarker$, $w_{|w|+1} = \rendmarker$. A configuration $c = (q, \alpha \times \vec{i}, j)$ has the following meaning. The head of 2DPPDA $M$ is over the symbol $w_j$ in the state $q$; the pushdown contains $\alpha = Z_mZ_{m-1}\cdots Z_0$ (the stack grows from right to left) and there is also additional information vector $\vec{i} = i_m, i_{m-1}, \ldots, i_0$, $i_k \in I$ such that $Z_k$ was pushed to the pushdown store when the head was over the $i_k$-th cell.     

The automaton's move is defined via the relation $\der$ as follows. Let $\delta(q, a, Z_n) = (q', \beta, d)$. The relation
\[ (q, Z_n\alpha \times i_n:\vec{i}, j) \der (q', \alpha' \times \vec{i}', j')\]
is defined according to the following case analysis.
\begin{itemize}    
     \item If $d \in \{\lto, \dto, \rto\}$, then $j' = j - 1$, $j' = j$, $j' = j + 1$ respectively. The cases $a = \lendmarker$, $d = {\lto}$ and $a = \rendmarker$, $d = {\rto}$ are forbidden.
    \item If $\beta = \eps$ and $d \in \{\lto, \dto, \rto\}$, then $\alpha' = \alpha$, $\vec{i}' = \vec{i}$
     \item If $\beta = \eps$ and $d = {\uto}$, then $\alpha' = \alpha$, $\vec{i}' = \vec{i}$, $j' = i_n$ 
     \item If $\beta = X_1\cdots X_k$, $k > 0$, then $\alpha' = \beta Z_n\alpha$,  $\vec{i}' = \underbrace{j':j':\cdots :j'}_k:i_n:\vec{i}$
\end{itemize}

The initial configuration is $(q_0, Z_0 \times 0, 0)$ and an accepting configuration is $(q_f, \eps \times (), |w| + 1)$, where $q_f \in F$ and by $()$ we have denoted the empty sequence of integers. I.e., $M$ reaches the right end marker~$\rendmarker$ empties the stack and finishes the computation in an accepting state. Formally, a word $w$ is accepted by $M$ if there exists a computational path from the initial configuration to an accepting one.

In the case of 1DPPDA (or just DPPDA), the moves $\lto$ are forbidden.

\subsection{Properties of DPPDA}\label{sect:DPPDA:properties}
Now we discuss the properties of the model and provide some shortcuts for the following needs. Note that each move of a 2DPPDA is either push- or pop-move due to the sake of convenience in the proofs (induction invariants are simpler). At the same time, in constructions, it is convenient to have right, left, and even stay moves that do not change the stack. So we add moves $\hrto$, $\hlto$, and $\hdto$  that are syntactic sugar for such moves.
So, when we write $\delta (q, a, z) = (p, \hrto)$, we mean the sequence of moves:
\[\delta (q, a, z) = (p', Z', \rto);\; \forall \sigma \in \Sigma_{\lendmarker\rendmarker} : \delta(p', \sigma, Z') = (p, \eps, \dto).\]
The construction for $\hlto$ and $\hdto$ are similar.

Due to the definition of $\delta$, a DPPDA can move only if the stack is non-empty and since each move is either push or pop, we have that $Z_0$ lies at the bottom of the stack till the last move of a computation or even after the last move in the case of unsuccessful computation. In the case of a successful computation, $Z_0$ is popped at the last move.

\begin{example} The following DPPDA $M$ recognizes the language $\{a^nb^nc^n \mid n \geq 1\}$. 
    \[M = \langle \{q_0, q_1, q_f\}, \{a,b,c, \lendmarker, \rendmarker\}, \{Z_0, Y, X\}, \{q_f\}, q_0, Z_0, \delta  \rangle, \]
where the transition function $\delta$ is defined as follows:
\begin{align*}
    \delta(q_0, \lendmarker, Z_0) &= (q_0, Y, \rto),
         & \delta(q_0, a, Y) &= (q_0, X, \rto),
         & \delta(q_0, a, X) &= (q_0, X, \rto),\\
    \delta(q_0, b, X) &= (q_0, \eps, \rto),
         & \delta(q_0, c, Y) &= (q_1, \eps, \uto)
         & \delta(q_1, a, Z_0) &= (q_1, \hrto)\\ 
    \delta(q_1, b, Z_0) &= (q_1, X, \rto),
        & \delta(q_1, b, X) &= (q_1, X, \rto),
        & \delta(q_1, c, X) &= (q_1, \eps, \rto),\\
    \delta(q_1, \rendmarker, Z_0) &= (q_f, \eps, \dto).        
\end{align*}

The demonstration of $M$ is provided on Fig.~\ref{Fig:AnBnCnEx}; 
on each subfigure we provide a configuration and the corresponding picture, note that $\delta$-rules are applied in the same order as we wrote them above.  We represent the stack via a linked list and attach the lists' nodes to the cell where the corresponding symbol had been pushed. The run of $M$ looks as follows (hereinafter $\ders$ is the reflexive and transitive closure of the relation $\der$, here we use it when omit some intermediate steps):
\begin{align*}
    & \big(q_0, Z_0\times0, 0\big)  \der \big(q_0, (YZ_0\times 1:0), 1\big) \der \big(q_0, (XYZ_0\times 2:1:0), 2\big) \ders \\
    \ders & \big(q_0, (XXXYZ_0\times 4:3:2:1:0), 4\big) \der
    \big(q_0, (XXYZ_0\times 3:2:1:0), 5\big) \ders \\
    \ders & \big(q_0, (YZ_0\times 1:0), 7\big) \der \big(q_1, (Z_0\times 0), 1\big) 
        \ders \big(q_1, (Z_0\times 0), 4\big) \der \big(q_1, (XZ_0\times 5:0), 5\big)\ders \\
    \ders & \big(q_1, (XXXZ_0\times 7:6:5:0), 7\big) \der \big(q_1, (XXZ_0\times 6:5:0), 8\big) \ders \\
    \ders & \big(q_0, (Z_0\times:0), 10\big) \der \big(q_f, (), 10)
\end{align*}

\newenvironment{AnBnCnTape}[0]{%
	\begin{tikzpicture}[ >=stealth]
        \useasboundingbox (-0.4,-1.25) rectangle (6.1,1.2);
	\edef\sizetape{0.6cm}
	\tikzstyle{tape}=[draw, minimum size=\sizetape]
	\tikzstyle{head}=[arrow box, draw, minimum size=.75cm, arrow box
	arrows={south:.2cm}, fill=lipicsYellow]

	\tikzstyle{arrup}=[single arrow, draw, shape border rotate=90, fill=gray!30, single arrow head extend=.001cm, single arrow head indent=.01ex]
	\tikzstyle{stack}=[rectangle split, rectangle split parts=##1,draw, anchor=center, minimum width = 1.5cm]
	\tikzstyle{stack_cell}=[rectangle, draw, anchor=center, minimum width = 1.5cm ]

	\begin{scope}[start chain=tape going right, node distance=-0.15mm]
	    \node [on chain=tape, tape] at (0, 0) (input) {\,$\lendmarker$\,};
                    \node [on chain=tape, tape] {$a$};
	                \node [on chain=tape, tape] {$a$};
            	    \node [on chain=tape, tape] {$a$};
            	    \node [on chain=tape, tape] {$b$};
            	    \node [on chain=tape, tape] {$b$};
                    \node [on chain=tape, tape] {$b$};
            	    \node [on chain=tape, tape] {$c$};
            	    \node [on chain=tape, tape] {$c$};
                    \node [on chain=tape, tape] {$c$};                                
		\node [on chain=tape, tape] {$\rendmarker$};		

	\end{scope}
	\newcommand{\head}[2]{\node [head, yshift=.375cm] at (tape-##1.north) (head##1) {##2};
	\node [yshift=-.5cm] at (head##1.south)  (head##1_cpoint) {};	
	}            		

	\newcommand{\ahline}[2]{
		\draw[->] (head##1.south) to [out=-60,in=-120] (Hanchor##2.south);
	}	
				
	\newcommand{\aaline}[2]{
		\draw[-, very thick] (Hanchor##1.south) to [out=-60,in=-120] (Hanchor##2.south);
    }

        \tikzstyle{Hanchor}=[rectangle, draw, minimum width=\sizetape, minimum height=15pt, fill=lipicsLightGray] 		
        
		\newcommand{\HanchorTape}[2]{\node [Hanchor, yshift=-.25cm] at (tape-##1.south) (Hanchor##1) {##2};
		\node [yshift=-.5cm] at (Hanchor##1.south)  (Hanchor##1_cpoint) {};	
		}
		\newcommand{\HanchorAnchor}[3]{\node [Hanchor, yshift=-.25cm] at (##1.south) (Hanchor##2) {##3};
		\node [yshift=-.5cm] at (Hanchor##2.south)  (Hanchor##2_cpoint) {};	
		}
        
   		\tikzstyle{arrup}=[single arrow, draw, shape border rotate=90, fill=gray!30, single arrow head extend=.001cm, single arrow head indent=.01ex]
        }%
        {%
    \end{tikzpicture}%
}
    
\begin{figure}
    \begin{subfigure}[t]{0.5\textwidth}
    \centering
        \begin{AnBnCnTape}                      
               \head{1}{$q_0$}
               \HanchorTape{1}{$Z_0$}
        \end{AnBnCnTape} 
        \vspace*{-15pt}
        \[\big(q_0, Z_0\times0, 0\big) \]
        \caption{Initial configuration.}
    \end{subfigure}
    \begin{subfigure}[t]{0.5\textwidth}
    \centering
        \begin{AnBnCnTape}                      
               \head{2}{$q_0$}
               \HanchorTape{1}{$Z_0$}
               \HanchorTape{2}{$Y$}
               \aaline{1}{2}
        \end{AnBnCnTape} 
        \vspace*{-15pt}
        \[\big(q_0, (YZ_0\times 1:0), 1\big)\]
        \caption{Push $Y$ to return back after checking the balance.}
    \end{subfigure}
    \begin{subfigure}[t]{0.5\textwidth}
    \centering
        \begin{AnBnCnTape}                      
               \head{3}{$q_0$}
               \HanchorTape{1}{$Z_0$}\aaline{1}{2}
               \HanchorTape{2}{$Y$}
                \HanchorTape{3}{$X$}\aaline{3}{2}
        \end{AnBnCnTape} 
        \vspace*{-10pt}
        \[\big(q_0, (XYZ_0\times 2:1:0), 2\big)\]
        \caption{Push $X$ while processing of $a$'s.}
    \end{subfigure}
    \begin{subfigure}[t]{0.5\textwidth}
    \centering
        \begin{AnBnCnTape}                      
               \head{5}{$q_0$}
               \HanchorTape{1}{$Z_0$}
               \HanchorTape{2}{$Y$}\aaline{1}{2}

               \HanchorTape{3}{$X$}\aaline{2}{3}               
               
               {\HanchorTape{4}{$X$}\aaline{4}{3}}               
               
               {\HanchorTape{5}{$X$}\aaline{5}{4}}
        \end{AnBnCnTape} 
        \vspace*{-10pt}
        \[\big(q_0, (XXXYZ_0\times 4:3:2:1:0), 4\big)\]
        \caption{Pop $X$ while processing of $b$'s.}
    \end{subfigure}
    \begin{subfigure}[t]{0.5\textwidth}
    \centering
        \begin{AnBnCnTape}                      
               \head{6}{$q_0$}
               \HanchorTape{1}{$Z_0$}
               \HanchorTape{2}{$Y$}\aaline{1}{2}

               \HanchorTape{3}{$X$}\aaline{2}{3}               
               
               {\HanchorTape{4}{$X$}\aaline{4}{3}}               
               
        \end{AnBnCnTape} 
        \vspace*{-10pt}
        \[\big(q_0, (XXYZ_0\times 3:2:1:0), 5\big)\]
        \caption{Pop $X$ while processing of $b$'s.}
    \end{subfigure}
    \begin{subfigure}[t]{0.5\textwidth}
    \centering
        \begin{AnBnCnTape}                      
               \head{8}{$q_0$}
               \HanchorTape{1}{$Z_0$}
               \HanchorTape{2}{$Y$}\aaline{1}{2}                                           
        \end{AnBnCnTape} 
        \vspace*{-10pt}
        \[\big(q_0, (YZ_0\times 1:0), 7\big)\]
        \caption{Check that $Y$ is on top when meet $c$ and pop $\uto$.}
    \end{subfigure}
    \begin{subfigure}[t]{0.5\textwidth}
    \centering
        \begin{AnBnCnTape}                      
               \head{2}{$q_1$}
               \HanchorTape{1}{$Z_0$}               
        \end{AnBnCnTape} 
        \vspace*{-20pt}
        \[\big(q_1, Z_0\times 0, 1\big)\]
        \caption{Ignore $a$'s via $\hrto$ moves.}
    \end{subfigure}
    \begin{subfigure}[t]{0.5\textwidth}
    \centering
        \begin{AnBnCnTape}                      
               \head{5}{$q_1$}
               \HanchorTape{1}{$Z_0$}               
        \end{AnBnCnTape} 
        \vspace*{-20pt}
        \[\big(q_1, Z_0\times 0, 4\big)\]
        \caption{Push $X$ while processing of $b$'s.}
    \end{subfigure}  
    \begin{subfigure}[t]{0.5\textwidth}
    \centering
        \begin{AnBnCnTape}                      
               \head{6}{$q_1$}
               \HanchorTape{1}{$Z_0$}               
               {\HanchorTape{6}{$X$}\draw[-, very thick] (Hanchor1.south) to [out=-20,in=-160] (Hanchor6.south);}
        \end{AnBnCnTape} 
        \vspace*{-10pt}
        \[\big(q_1, (XZ_0\times 5:0), 5\big)\]
        \caption{Push $X$ while processing of $b$'s.}
    \end{subfigure}
    \begin{subfigure}[t]{0.5\textwidth}
    \centering
        \begin{AnBnCnTape}                      
               \head{8}{$q_1$}
               \HanchorTape{1}{$Z_0$}               
               {\HanchorTape{6}{$X$}\draw[-, very thick] (Hanchor1.south) to [out=-20,in=-160] (Hanchor6.south);}
               \HanchorTape{7}{$X$}\aaline{6}{7}
               \HanchorTape{8}{$X$}\aaline{8}{7}
        \end{AnBnCnTape} 
        \vspace*{-10pt}
        \[\big(q_1, (XXXZ_0\times 7:6:5:0), 7\big)\]
        \caption{Pop $X$ while processing of $c$'s.}
    \end{subfigure}
    \begin{subfigure}[t]{0.5\textwidth}
    \centering
        \begin{AnBnCnTape}                      
               \head{9}{$q_1$}
               \HanchorTape{1}{$Z_0$}               
               {\HanchorTape{6}{$X$}\draw[-, very thick] (Hanchor1.south) to [out=-20,in=-160] (Hanchor6.south);}
               \HanchorTape{7}{$X$}\aaline{6}{7}
        \end{AnBnCnTape}
        \vspace*{-10pt} 
        \[\big(q_1, (XXZ_0\times 6:5:0), 8\big)\]
        \caption{Pop $X$ while processing of $c$'s.}
    \end{subfigure}
    \begin{subfigure}[t]{0.5\textwidth}
    \centering
        \begin{AnBnCnTape}                      
               \head{11}{$q_1$}
               \HanchorTape{1}{$Z_0$}               
        \end{AnBnCnTape} 
        \vspace*{-10pt}
        \[\big(q_0, Z_0\times0, 10\big)\]
        \caption{Check that $Z_0$ is on top when meet $\rendmarker$ and pop $\dto$.}
    \end{subfigure}
    \caption{Demonstration of DPPDA recognizing $a^nb^nc^n$.}\label{Fig:AnBnCnEx}
\end{figure}    
    
\end{example}

\section{Equivalence of DPPDAs and PEGs}

In this section, we provide an algorithm that transforms a PEG into a DPPDA and vise versa.
Our construction is similar to the well-known proof of equivalence between CFGs and DPDA for CFLs, but since both DPPDAs and PEGs are more complicated than DPDAs and CFLs, our constructions are technically harder. We refer the reader to~\cite{Sipser13} for the detailed explanation of the proof idea, where it was provided for CFLs (so as the proof for CFLs as well).

\subsection{PEG to DPPDA}

In this section we assume that PEGs have a special form. We call it Chomsky's normal form since it is similar to such a form for CFGs.


\begin{definition} 
    A PEG $G$ has a \emph{Chomsky normal form} if the axiom $S$ never occurs on the right side of the rules and the rules are of the following form:
    \[ A \lto B \por C,\quad  A \lto BC,\quad A \lto \pnot B,\quad  A \lto a,\quad A \lto \eps.\]
\end{definition}

\begin{lemma}\label{Lemma:CNF}
    Each PEG $G$ has an equivalent PEG $G'$ in Chomsky's normal form which is complete if so was $G$. 
\end{lemma}

The proof of the lemma is straight-forward and uses almost the same algorithm as for the transformation of a CF-grammar to a grammar in the Chomsky normal form (see e.g.,~\cite{Hopcroft1979}), so we describe only a proof idea.

\begin{proof}[Proof idea of Lemma~\ref{Lemma:CNF}]
     To get rid of long concatenations we replace long expressions with their shortcuts, e.g., concatenation $ABC$ is replaced by a single nonterminal $[ABC]$ and then rules  $[ABC] \lto A[BC]$  and $[BC] \lto BC$ are added. A similar transformation works with longer concatenation and ordered choice. For negation we use the similar trick: We replace $\pnot(\mathrm{expr}) $ by $ [\pnot(\mathrm{expr})]$ and add the rule  $[\pnot(\mathrm{expr})] \lto  \pnot[(\mathrm{expr})]$.
\end{proof}

\begin{theorem} 
    For a PEG $G$ there exists an equivalent DPPDA $M$.    
\end{theorem}
\begin{proof}
    We assume that $G$ is a well-formed PEG in a Chomsky normal form (by Lemma~\ref{Lemma:CNF}). We construct an equivalent DPPDA $M = \langle Q, \Sigma_{\lendmarker\rendmarker}, \Gamma, \{q_f\}, q_0, Z_0, \delta  \rangle$ by the PEGs description. We formally describe $\delta$ on Fig.~\ref{Fig:PEGtoDPPDA}; we do not provide a full list of states $Q$ and pushdown alphabet $\Gamma$ since most of the states and symbols depend on rules listed in $\delta$'s construction and can be easily restored from it. Since the construction is straightforward, we describe here only the main details. 

\begin{figure}
\begin{center}
    \setlength{\fboxsep}{2ex}
    \fbox{
    \begin{minipage}{0.9\textwidth}
    \centerline{\textbf{{Construction of $\delta$}}}
    \medskip
    We denote by $Z \in \Gamma$ and $\sigma \in \Sigma_{\lendmarker\rendmarker}$ arbitrary symbols. The rules are grouped with respect to the PEG's operations. Note that the states $q$, $q_{A_{\pm}}$ are the same for all rules, while other states and stack symbols depend on the rule, i.e., stack symbols $A_1$'s from different rules are different even if they correspond to the same nonterminal $A$. When we use nonterminals (and states) with signs $\pm$ or $\mp$, the signs have corresponding matching, i.e., if in a rule we have $A_{\pm}$ and $B_{\mp}$, then when $A_{\pm}=A_{+}$, $B_{\mp}$ equals to $B_{-}$ and when $A_{\pm}=A_{-}$, $B_{\mp}$ equals to $B_{+}$.        
    \begin{enumerate}
        \begin{multicols}{2} 
        \setcounter{enumi}{-1}
        \item General rules
        \begin{itemize}[topsep=0pt]
            \item $\delta(q_0, \lendmarker, Z_0) = (q, S, \rto)$;
            \item $\delta(q_{A_{\pm}}, \sigma, A) = (q_{A_{\pm}}, \eps, \dto)$;
            \item $\delta(q_{S_{+}}, \rendmarker, Z_0) = (q_f, \eps, \dto)$.
        \end{itemize}
        \item $A \lto BC$
        \begin{itemize}[topsep=0pt, itemsep=1ex]
            \item $\delta(q, \sigma, A) = (q, BA_1, \dto)$;
            \item $\delta(q_{B_{+}}, \sigma, A_{1}) = (q, CA_2, \dto)$; 
            \item $\delta(q_{B_{-}}, \sigma, A_{1}) = (q_{A_{-}}, \eps, \uto)$;
            \item $\delta(q_{C_{+}}, \sigma, A_{2}) = (q_{A_2}, \eps, \dto)$;
            \item $\delta(q_{A_2}, \sigma, A_1) = (q_{A_{+}}, \eps, \dto)$;
            \item $\delta(q_{C_{-}}, \sigma, A_{2}) = (q_{A_{2-}}, \eps, \uto)$;
            \item $\delta(q_{A_{2-}}, \sigma, A_1) =  (q_{A_{-}}, \eps, \uto) $.
        \end{itemize} 
        \columnbreak
        \item $A \lto B \por C$
        \begin{itemize}[topsep=0pt, itemsep=1ex]
            \item $\delta(q, \sigma, A) = (q, BA_1, \dto)$;
            \item 
                  $\delta(q_{B_{+}}, \sigma, A_{1}) = (q_{A_{+}}, \eps, \dto)$;
            \item 
                  $\delta(q_{B_{-}}, \sigma, A_{1}) = (q_{A_2}, \eps, \uto)$;
            \item $\delta(q_{A_2}, \sigma, Z) = (q, CA_2, \dto)$;
            \item $\delta(q_{C_{+}}, \sigma, A_{2}) = (q_{A_{+}}, \eps, \dto)$;
            \item $\delta(q_{C_{-}}, \sigma, A_{2}) = (q_{A_{-}}, \eps, \uto)$.
        \end{itemize}
        \item $A \lto \eps$ 
        \begin{itemize}[topsep=0pt]
            \item $\delta(q, \sigma, A) = (q_{A_{+}}, \hdto)$.
        \end{itemize}
        
        \item $A \lto \pnot B$ 
        \begin{itemize}[topsep=0pt, itemsep=1ex]
            \item $\delta(q, \sigma, A) = (q, BA_1, \dto)$;
            \item $\delta(q_{B_{\pm}}, \sigma, A_1) = (q_{A_{\mp}}, \eps, \uto)$.
        \end{itemize}
        \item $A \lto a$
        \begin{itemize}[topsep=0pt, itemsep=1ex]
            \item $\delta(q, a, A) = (q_{A_{+}}, \hrto)$;                      
            \item $\delta(q, b, A) = (q_{A_{-}}, \hdto)$, here $b \neq a$.
        \end{itemize}        
        \end{multicols}       
    \end{enumerate}
    \end{minipage} }
\end{center}    
    \caption{Construction of $\delta$ by the PEG $G$}\label{Fig:PEGtoDPPDA}
\end{figure}

The DPPDA $M$ simulates the parsing process of a PEG $G$ on the input $w$. Firstly $M$ performs the move $\delta(q_0, \lendmarker, Z_0) = (q, S, \rto)$ that leads it to the initial simulation configuration:  
$(q_0, Z_0 \times 0, 0) \der (q, SZ_0 \times (1:0), 1)$, 
where $q$ is the \emph{main work state} and $S$ is the axiom of the PEG. 


During the simulation the following invariants hold. Below $A$ is a nonterminal of the PEG.

\begin{enumerate}
    \item\label{inv:qA} If the automaton is in the main work state $q$ and on the top of the stack is the pair $A \times i$, then the head is over the cell $i$.
    \item\label{inv:Subw} If the head is over the cell $r$ in a state $q_{A_{\pm}}$ (hereinafter $q_{A_{\pm}} \in \{q_{A_+}, q_{A_-}\}$) and the topmost symbol had been added at the position $l$, then it means the following.
    \begin{enumerate}
        \item[$q_{A_{+}}$]\label{inv:Subw:success} A subword $s = w_l\cdots w_{r-1}$, $r \geq l$, would be parsed by the PEG from $A$ (starting from the position $l$); when $r = l$, we have $s = \eps$. In the other direction: if the PEG parses $w_l\cdots w_{r-1}$ from $A$ starting from the position $l$, then the DPPDA that starts computation from the position $l$ in the main work state $q$ with $A$ on the top on the stack finishes at the position $r$ with (the same) $A$ on the top of the stack, i.e.,
        \[(q, A\alpha\times l:\vec{i}, l) \ders (q_{A_{+}}, A\alpha\times l:\vec{i}, r). \]
        \item[$q_{A_{-}}$]\label{inv:Subw:failure} After the PEG started parsing from $A$ from the position $l$, the computation ended up with a failure at some point in the case of $q_{A_{-}}$. In the other direction: if the PEG fails, then for some $r \geq l$: 
         \[(q, A\alpha\times l:\vec{i}, l) \ders (q_{A_{-}}, A\alpha\times l:\vec{i}, r). \]
    \end{enumerate}         
\end{enumerate}

DPPDA $M$ accepts the input only if the head reaches the symbol $\rendmarker$ in the state $q_{S+}$ (note that the axiom does not occur on the right side of the rules). 
Formally, we add the rule \[\delta(q_{S_{+}}, \rendmarker, Z_0) = (q_f, \eps, \dto),\]
where $q_f$ is the only final state of the DPPDA. So, from the invariant it follows that the DPPDA accepts the input iff the PEG parses the input.

The rest of the construction is the delta's description in Fig.~\ref{Fig:PEGtoDPPDA}. The proof is a straightforward induction on the recursion depth of the PEGs computation. So we describe the behavior of the automaton corresponding to the formal construction in two main cases and check that the invariants hold (the remaining cases are simple).

Each rule is applied to a configuration of the form \[(q, A\alpha \times l:\vec{i}, l).\]

The first case (of concatenation) is illustrated on Fig~\ref{Fig:ConcatRuns}. The automaton works as follows.
It pushes the auxiliary symbol $A_1$ at the same position that $A$ has been pushed (since the invariant~\ref{inv:qA} holds) and then pushes $B$. 
If it reaches a configuration of the form \[(q_{B_{+}}, BA_1A\alpha \times l:l:l:\vec{i}, r'),\] then $B$ has successfully parsed the subword $w_l\cdots w_{r'-1}$ due to the invariant~\ref{inv:Subw:success}, then $B$ is popped due to General rules and DPPDA pushes $C$ at the position $r'$ and goes to the main work state $q$.
If then the DPPDA reaches a configuration of the form 
\[(q_{C_{+}}, CA_2A_1A\alpha\times r':r':l:l:\vec{i}, r),\]
we have that the PEG parsed the word~$w_{r'}\cdots w_{r-1}$ from $C$ and after the sequences of technical pops the automaton comes to the configuration $(q_{A_{+}}, A\alpha \times l:\vec{i}, r)$ that proves that the invariant~\ref{inv:Subw}-$q_{A_{+}}$ holds (the arguments for the other direction are similar). 

\newenvironment{ConcatTape}[0]{%
	\begin{tikzpicture}[ >=stealth]
        \useasboundingbox (-0.4,-2) rectangle (6.1,1.2);
	\edef\sizetape{0.7cm}
	\tikzstyle{tape}=[draw, minimum size=\sizetape]
	\tikzstyle{head}=[arrow box, draw, minimum size=.75cm, arrow box
	arrows={south:.2cm}, fill=lipicsYellow]
    
	\tikzstyle{arrup}=[single arrow, draw, shape border rotate=90, fill=gray!30, single arrow head extend=.001cm, single arrow head indent=.01ex]
	\tikzstyle{stack}=[rectangle split, rectangle split parts=##1,draw, anchor=center, minimum width = 1.5cm]
	\tikzstyle{stack_cell}=[rectangle, draw, anchor=center, minimum width = 1.5cm ]

	\begin{scope}[start chain=tape going right, node distance=-0.15mm]
	    \node [on chain=tape, tape] at (0, 0) (input) {$\lendmarker$};
                    \node [on chain=tape, tape] {$\color{lipicsLineGray}1$};
                    \node [on chain=tape, tape] {\;$\cdots$\;};
            	    \node [on chain=tape, tape] {$\color{lipicsLineGray}l$}; 
            	    \node [on chain=tape, tape] {$\cdots$};
                    \node [on chain=tape, tape] {$\color{lipicsLineGray}r'$}; 
            	    \node [on chain=tape, tape] {$\cdots$};
                    \node [on chain=tape, tape] {$\color{lipicsLineGray}r$};  
		\node [on chain=tape, tape] {$\rendmarker$};		

	\end{scope}
    \edef\posl{4}                
    \edef\posRb{6}
    \edef\posRc{8}                
	\newcommand{\head}[2]{\node [head, yshift=.38cm] at (tape-##1.north) (head##1) {##2};
	\node [yshift=-.5cm] at (head##1.south)  (head##1_cpoint) {};	
	}            		

	\newcommand{\ahline}[2]{
		\draw[->] (head##1.south) to [out=-60,in=-120] (Hanchor##2.south);
	}	
				
	\newcommand{\aaline}[2]{
		\draw[-, very thick] (Hanchor##1.south) to [out=-60,in=-120] (Hanchor##2.south);
    }

   		\tikzstyle{Hanchor}=[rectangle, draw, minimum width=.7cm, minimum height=15pt, fill=lipicsLightGray] 
        
		\newcommand{\HanchorTape}[2]{\node [Hanchor, yshift=-.25cm] at (tape-##1.south) (Hanchor##1) {##2};
		\node [yshift=-.5cm] at (Hanchor##1.south)  (Hanchor##1_cpoint) {};	
		}
		\newcommand{\HanchorAnchor}[3]{\node [Hanchor, yshift=-.25cm] at (Hanchor##1.south) (Hanchor##2) {##3};
		\node [yshift=-.5cm] at (Hanchor##2.south)  (Hanchor##2_cpoint) {};	
		}
        
   		\tikzstyle{arrup}=[single arrow, draw, shape border rotate=90, fill=gray!30, single arrow head extend=.001cm, single arrow head indent=.01ex]
        \HanchorTape{1}{$Z_0$}
        \HanchorTape{2}{$S$}
        \node [yshift=-1cm] at (tape-3.south) (Zdots) {\!$\cdots$\;\;};
        \draw[-,very thick] (Hanchor2.south) to [out=-60,in=-180] ([xshift=1pt]Zdots.west);
        \aaline{1}{2}
        \newcommand{\Astack}[0]{\HanchorTape{\posl}{$A$}
               \draw[-,very thick] ([xshift=-10pt]Zdots.east) to [out=0,in=-180] (Hanchor\posl.west);}
        
        }%
        {%
    \end{tikzpicture}%
}

\begin{figure}
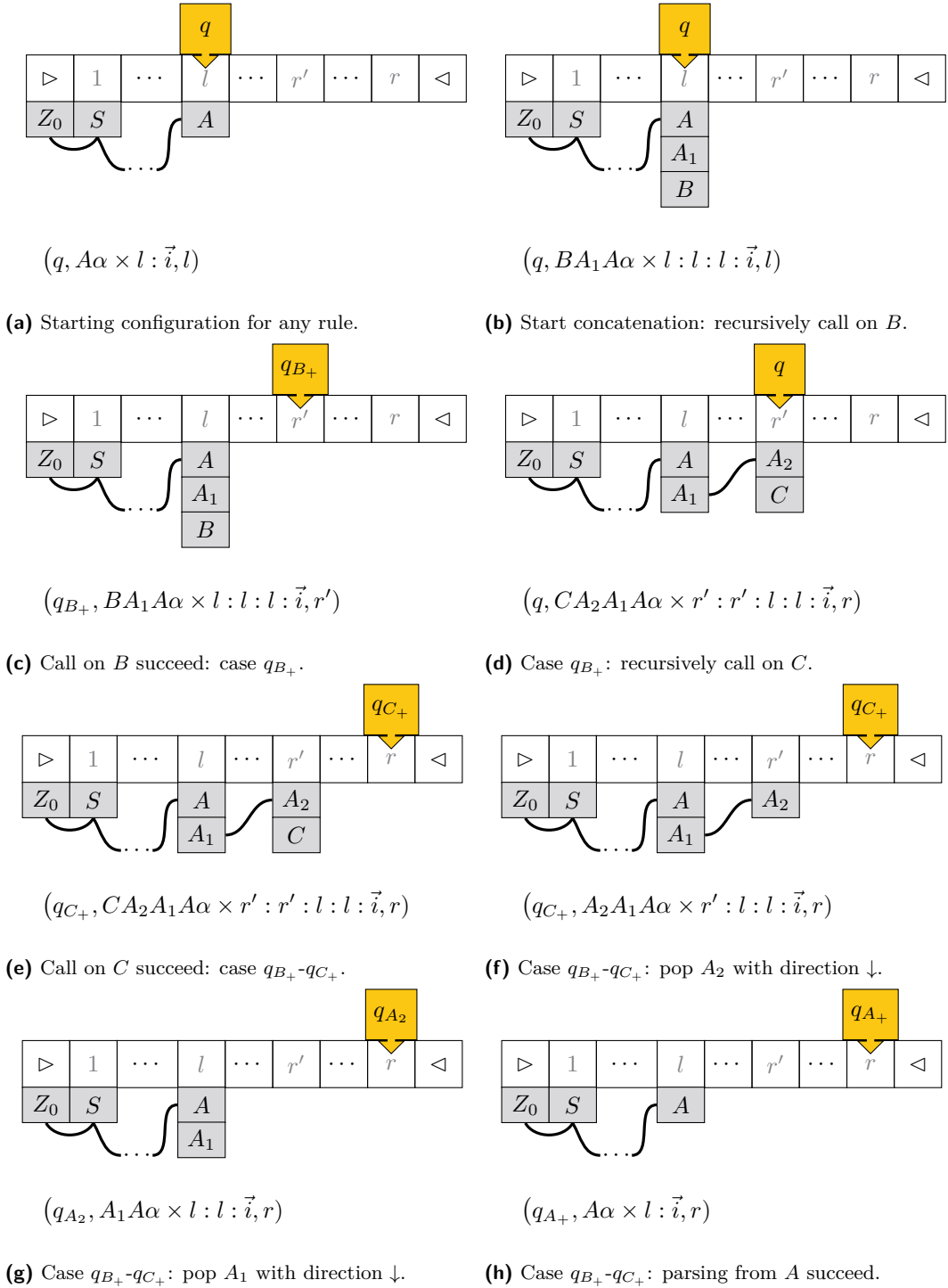

    \begin{subfigure}[t]{0.5\textwidth}
    \centering
        \begin{ConcatTape} 
               \head{\posl}{$q$}
               \Astack
        \end{ConcatTape} 
        \[\big(q, A\alpha \times l:\vec{i}, l) \]
        \caption{Starting configuration for any rule.}
    \end{subfigure}
    \begin{subfigure}[t]{0.5\textwidth}
    \centering
        \begin{ConcatTape}  
               \head{\posl}{$q$}
               \Astack
               \HanchorAnchor{\posl}{A1}{$A_1$}
               \HanchorAnchor{A1}{B}{$B$}
        \end{ConcatTape} 
        \[\big(q, BA_1A\alpha \times l:l:l:\vec{i}, l) \]
        \caption{Start concatenation: recursively call on $B$.}
    \end{subfigure}
    \begin{subfigure}[t]{0.5\textwidth}
    \centering
        \begin{ConcatTape}                      
               \head{\posRb}{$q_{B_{+}}$}
               \Astack
               \HanchorAnchor{\posl}{A1}{$A_1$}
               \HanchorAnchor{A1}{B}{$B$}
        \end{ConcatTape} 
        \[\big(q_{B_{+}}, BA_1A\alpha \times l:l:l:\vec{i}, r') \]
        \caption{Call on $B$ succeed: case $q_{B_{+}}$.}
    \end{subfigure}    
    \begin{subfigure}[t]{0.5\textwidth}
    \centering
        \begin{ConcatTape}                      
               \head{\posRb}{$q$}
               \Astack
               \HanchorAnchor{\posl}{A1}{$A_1$}               
               \HanchorTape{\posRb}{$A_2$} \draw[-,very thick] (HanchorA1.east) to [out=0,in=-180] (Hanchor\posRb.west); 
               \HanchorAnchor{\posRb}{C}{$C$}

        \end{ConcatTape} 
        \[\big(q, CA_2A_1A\alpha \times r':r':l:l:\vec{i}, r) \]
        \caption{Case $q_{B_{+}}$: recursively call on $C$.}
    \end{subfigure}
    \begin{subfigure}[t]{0.5\textwidth}
    \centering
        \begin{ConcatTape}                      
               \head{\posRc}{$q_{C_{+}}$}
               \Astack
               \HanchorAnchor{\posl}{A1}{$A_1$}
               \HanchorTape{\posRb}{$A_2$} \draw[-,very thick] (HanchorA1.east) to [out=0,in=-180] (Hanchor\posRb.west); 
               \HanchorAnchor{\posRb}{C}{$C$}

        \end{ConcatTape} 
        \vspace*{-15pt}
        \[\big(q_{C_{+}}, CA_2A_1A\alpha \times r':r':l:l:\vec{i}, r) \]
        \caption{Call on $C$ succeed: case $q_{B_{+}}$-$q_{C_{+}}$.}
    \end{subfigure}
    \begin{subfigure}[t]{0.5\textwidth}
    \centering
        \begin{ConcatTape}                      
               \head{\posRc}{$q_{C_{+}}$}
               \Astack
               \HanchorAnchor{\posl}{A1}{$A_1$}
               \HanchorTape{\posRb}{$A_2$} \draw[-,very thick] (HanchorA1.east) to [out=0,in=-180] (Hanchor\posRb.west);
        \end{ConcatTape} 
        \vspace*{-15pt}
        \[\big(q_{C_{+}}, A_2A_1A\alpha \times r':l:l:\vec{i}, r) \]
        \caption{Case $q_{B_{+}}$-$q_{C_{+}}$: pop $A_2$ with direction $\dto$.}
    \end{subfigure}
    \begin{subfigure}[t]{0.5\textwidth}
    \centering
        \begin{ConcatTape}                      
               \head{\posRc}{$q_{A_2}$}
               \Astack
               \HanchorAnchor{\posl}{A1}{$A_1$}
        \end{ConcatTape} 
        \vspace*{-15pt}
        \[\big(q_{A_2}, A_1A\alpha \times l:l:\vec{i}, r) \]
        \caption{Case $q_{B_{+}}$-$q_{C_{+}}$: pop $A_1$ with direction $\dto$.}
    \end{subfigure}
    \begin{subfigure}[t]{0.5\textwidth}
    \centering
        \begin{ConcatTape}                      
               \head{\posRc}{$q_{A_{+}}$}
               \Astack
        \end{ConcatTape} 
        \vspace*{-15pt}
        \[\big(q_{A_{+}}, A\alpha \times l:\vec{i}, r) \]
        \caption{Case $q_{B_{+}}$-$q_{C_{+}}$: parsing from $A$ succeed.}
    \end{subfigure}
        \caption{Illustration of Concatenation}\label{Fig:ConcatRuns}
\end{figure}

In the case of reaching either the configuration 
\[(q_{C_{-}}, CA_2A_1A\alpha \times r':r':l:l:\vec{i}, r) \text { or } (q_{B_{-}}, BA_1A\alpha \times l: l : l:\vec{i}, r')\]
earlier, the sequence of pops leads the DPPDA to the configuration $(q_{A_{-}}, A\alpha \times l:\vec{i}, l)$ that proves that the invariant~\ref{inv:Subw}-$q_{A_{-}}$ holds (the arguments for the other direction are similar).


The case of the ordered choice is similar to the case of concatenation. The difference between the cases is as follows. In the case of a configuration \[(q_{B_{+}}, BA_1A\alpha \times l:l:l:\vec{i}, r),\]
the automaton reaches the configuration $(q_{A_{+}}, A\alpha \times l:\vec{i}, r)$ via the technical moves, so this case is trivial. Consider the case of a configuration \[(q_{B_{-}}, BA_1A\alpha \times l:l:l:\vec{i}, r).\]
We illustrate it on Fig.~\ref{Fig:OrderedChoiceDemo}; we begin the demonstration after the push of $B$ since the steps before that are the same as for the concatenation (Fig.~\ref{Fig:ConcatRuns} (a) and (b), provided $r'=r$).
The automaton reaches the configuration $(q, CA_2A\alpha \times l:l:l:\vec{i}, l)$ after which for some $r$
\[\text{either } (q, CA_2A\alpha \times l:l:l:\vec{i}, l) 
\ders (q_{C_{+}}, CA_2A\alpha \times l:l:l:\vec{i}, r) \ders
(q_{A_{+}}, A\alpha \times l:\vec{i}, r),\]
\[\text{or } (q, CA_2A\alpha \times l:l:l:\vec{i}, l) \ders (q_{C_{-}}, CA_2A\alpha \times l:l:l:\vec{i}, r) \ders (q_{A_{-}}, A\alpha \times l:\vec{i}, l).\]


\newenvironment{OrderedChoiceTape}[0]{%
	\begin{tikzpicture}[ >=stealth]
        \useasboundingbox (-0.4,-2) rectangle (6.1,1.2);
	\edef\sizetape{0.7cm}
	\tikzstyle{tape}=[draw, minimum size=\sizetape]
	\tikzstyle{head}=[arrow box, draw, minimum size=.75cm, arrow box
	arrows={south:.2cm}, fill=lipicsYellow]
    
	\tikzstyle{arrup}=[single arrow, draw, shape border rotate=90, fill=gray!30, single arrow head extend=.001cm, single arrow head indent=.01ex]
	\tikzstyle{stack}=[rectangle split, rectangle split parts=##1,draw, anchor=center, minimum width = 1.5cm]
	\tikzstyle{stack_cell}=[rectangle, draw, anchor=center, minimum width = 1.5cm ]

	\begin{scope}[start chain=tape going right, node distance=-0.15mm]
	    \node [on chain=tape, tape] at (0, 0) (input) {$\lendmarker$};
                    \node [on chain=tape, tape] {$\color{lipicsLineGray}1$};
                    \node [on chain=tape, tape] {\;$\cdots$\;};
            	    \node [on chain=tape, tape] {$\color{lipicsLineGray}l$}; 
            	    \node [on chain=tape, tape] {$\cdots$};
                    \node [on chain=tape, tape] {$\color{lipicsLineGray}r$}; 
            	    \node [on chain=tape, tape] {$\cdots$};
		\node [on chain=tape, tape] {$\rendmarker$};		

	\end{scope}
    \edef\posl{4}                
    \edef\posRb{6}
    \edef\posRc{8}                
	\newcommand{\head}[2]{\node [head, yshift=.38cm] at (tape-##1.north) (head##1) {##2};
	\node [yshift=-.5cm] at (head##1.south)  (head##1_cpoint) {};	
	}            		

	\newcommand{\ahline}[2]{
		\draw[->] (head##1.south) to [out=-60,in=-120] (Hanchor##2.south);
	}	
				
	\newcommand{\aaline}[2]{
		\draw[-, very thick] (Hanchor##1.south) to [out=-60,in=-120] (Hanchor##2.south);
    }

   		\tikzstyle{Hanchor}=[rectangle, draw, minimum width=.7cm, minimum height=15pt, fill=lipicsLightGray] 
        
		\newcommand{\HanchorTape}[2]{\node [Hanchor, yshift=-.25cm] at (tape-##1.south) (Hanchor##1) {##2};
		\node [yshift=-.5cm] at (Hanchor##1.south)  (Hanchor##1_cpoint) {};	
		}
		\newcommand{\HanchorAnchor}[3]{\node [Hanchor, yshift=-.25cm] at (Hanchor##1.south) (Hanchor##2) {##3};
		\node [yshift=-.5cm] at (Hanchor##2.south)  (Hanchor##2_cpoint) {};	
		}
        
   		\tikzstyle{arrup}=[single arrow, draw, shape border rotate=90, fill=gray!30, single arrow head extend=.001cm, single arrow head indent=.01ex]
        \HanchorTape{1}{$Z_0$}
        \HanchorTape{2}{$S$}
        \node [yshift=-1cm] at (tape-3.south) (Zdots) {\!$\cdots$\;\;};
        \draw[-,very thick] (Hanchor2.south) to [out=-60,in=-180] ([xshift=1pt]Zdots.west);
        \aaline{1}{2}
        \newcommand{\Astack}[0]{\HanchorTape{\posl}{$A$}
               \draw[-,very thick] ([xshift=-10pt]Zdots.east) to [out=0,in=-180] (Hanchor\posl.west);}
        
        }%
        {%
    \end{tikzpicture}%
}

\begin{figure}
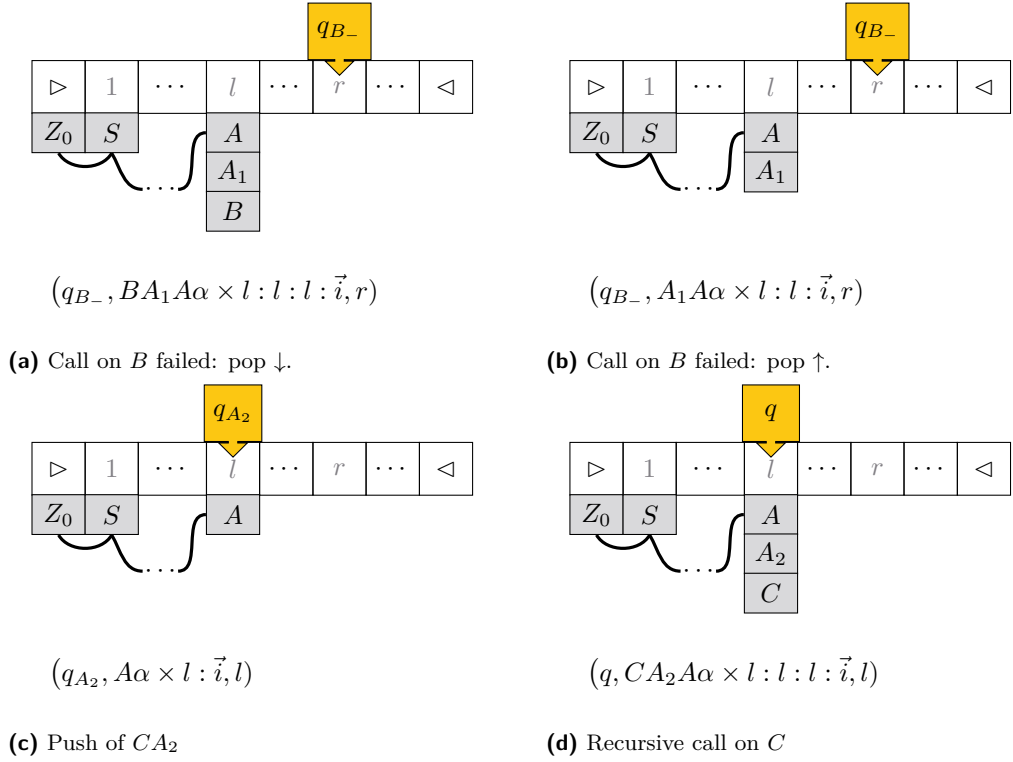

    \begin{subfigure}[t]{0.5\textwidth}
    \centering
        \begin{OrderedChoiceTape}                      
               \head{\posRb}{$q_{B_{-}}$}
               \Astack
               \HanchorAnchor{\posl}{A1}{$A_1$}
               \HanchorAnchor{A1}{B}{$B$}
        \end{OrderedChoiceTape} 
        \[\big(q_{B_{-}}, BA_1A\alpha \times l:l:l:\vec{i}, r) \]
        \caption{Call on $B$ failed: pop $\dto$.}
    \end{subfigure}   
    \begin{subfigure}[t]{0.5\textwidth}
    \centering
        \begin{OrderedChoiceTape}                      
               \head{\posRb}{$q_{B_{-}}$}
               \Astack
               \HanchorAnchor{\posl}{A1}{$A_1$}
        \end{OrderedChoiceTape} 
        \[\big(q_{B_{-}}, A_1A\alpha \times l:l:\vec{i}, r) \]
        \caption{Call on $B$ failed: pop $\uto$.}
    \end{subfigure} 
    \begin{subfigure}[t]{0.5\textwidth}
    \centering
        \begin{OrderedChoiceTape}                      
               \head{\posl}{$q_{A_2}$}
               \Astack
        \end{OrderedChoiceTape} 
        \[\big(q_{A_{2}}, A\alpha \times l:\vec{i}, l) \]
        \caption{Push of $CA_2$}
    \end{subfigure}
    \begin{subfigure}[t]{0.5\textwidth}
    \centering
        \begin{OrderedChoiceTape}                      
               \head{\posl}{$q$}
               \Astack
               \HanchorAnchor{\posl}{A2}{$A_2$}
               \HanchorAnchor{A2}{C}{$C$}
        \end{OrderedChoiceTape} 
        \[\big(q, CA_2A\alpha \times l:l:l:\vec{i}, l) \]
        \caption{Recursive call on $C$}
    \end{subfigure}
    \caption{Illustration of Ordered choice.}\label{Fig:OrderedChoiceDemo}
\end{figure}

The analysis of the remaining cases directly follows from the definitions, so we omit it.
\end{proof}

\subsection*{DPPDA to PEG}

In this subsection, we need DPPDA of a special form for the sake of construction.

\begin{definition}
 Consider (part of) a run $(q, Z \times i, j) \ders (p, (), k)$ at which $Z$ was finally popped. We say that the \emph{pop direction} of $Z$ is $\uto$, $\dto$, $\lto$, $\rto$ depending on the last move's direction and denote it by $d_{q,j}(Z)$ or by $d(Z)$ if the run is fixed and there is no ambiguity.
\end{definition}

Note that the pop direction does not depend on the position $i$ at which $Z$ has been pushed.

\begin{lemma}\label{lemma:DPPDAspecProps}
For each DPPDA there exists an equivalent DPPDA for which the following properties hold.

\begin{enumerate}
    \item\label{prop:PushUD} In the case of pop only moves $\dto, \uto$ are allowed.
    \item\label{prop:PushOneByOne} In the case of push only one symbol is added to the stack. 
    \item\label{prop:Zzero} The bottom marker $Z_0$ remains in the stack until the last step and never occurs on other positions.  
    \item\label{prop:ZzeroFirstSymbol} Without loss of generality, we assume that the bottom marker~$Z_0$ pushed at the position $1$ (at the first input symbol, but not on the left end marker~$\lendmarker$).
    \item\label{prop:LastMoveDown} The last pop direction is $\dto$.    
\end{enumerate}
        
\end{lemma}

\begin{proof}
    For Property~\ref{prop:PushUD} we need the syntactic sugar moves $\hrto$ described in Subsection~\ref{sect:DPPDA:properties}. A pop move $\rto$ can be replaced by a series of moves $\hrto$ and $\dto$, where the pop operations are performed only at moves $\dto$.
    
    Now we assume that Property~\ref{prop:PushUD} holds and describe how to transform a DPPDA to achieve Property~\ref{prop:PushOneByOne}. If $k > 1$ symbols are pushed during a single move $\rto$, we replace them by a single push during the move $\rto$ and push of $k-1$ symbols by the following move $\dto$. If $k$ symbols are pushed during the move~$\dto$, we replace this move with a series of $k$ moves $\dto$ each of which pushes the corresponding symbol.  

   As discussed in Subsection~\ref{sect:DPPDA:properties}, the last move of any successful computation is the pop of $Z_0$. If Properties~\ref{prop:Zzero} and~\ref{prop:LastMoveDown} do not hold together for a DPPDA $M$ with the initial symbol in the pushdown storage $Z_0$, we construct another DPPDA $M'$ with the initial symbol in the pushdown storage $Z_0 \not\in\Gamma_M$ that pushes $Z_0$ at the very first move $\dto$, then simulates $M$ and pops $Z_0$ at each final state of $M$ with the direction $\dto$ which does not change the position of the head that should be on the right end-marker~$\rendmarker$ in the case of an accepting computation. So $M'$ satisfies Properties~\ref{prop:Zzero} and~\ref{prop:LastMoveDown}.   
  
  It is left to prove Property~\ref{prop:ZzeroFirstSymbol}. Due to Properties~\ref{prop:Zzero} and~\ref{prop:LastMoveDown}, the position of adding $Z_0$ to the stack does not matter. As before, we construct an equivalent DPPDA $M'$ from~$M$. If $M$ pushes a symbol $X$ when the head is over $\lendmarker$, $M'$ pushes a symbol $[X\lendmarker]$ to the stack. So, if $M$ arrives at $\lendmarker$ after the pop of $X$, $M'$ simulates this move via the finite control: if it pops $\uto$ a symbol of the form $[X\lendmarker]$, it behaves like the head is over $\lendmarker$. 
\end{proof}

For the rest of the section, we fix a DPPDA $M$ satisfying conditions of Lemma~\ref{lemma:DPPDAspecProps}. We construct an equivalent PEG for $M$ and prove that the PEG generates the languages $L(M)$.

Now we describe the PEG. Non-terminals have the following form: $[qZp\udto]$, $[qZp\uto]$, $[qZp\dto]$, and $[qZ\bar p]$, that have the following meaning. After pushing $Z$ to the stack, $M$ has the state $q$ and when (that) $Z$ is popped $M$ has the state $p$; $\uto$ and $\dto$ indicate the pop direction (on arriving at $p$). Nonterminals $[qZ\bar p]$ are auxiliary and nonterminals $[qZp\udto]$ denote any pop direction, formally $[qZp\udto] \lto [qZp\dto] \por [qZp\uto]$, the order of rules does not matter in our construction.

Below we describe the rules of PEG depending on $\delta(q, a, Z)$ and use the following convention. 
A state $s$ runs all the possible values of the $M$'s states, so the expression
$[qZp\dto] \lto a[rXs\udto][sZp\dto]$
is a shortcut for the rule \[[qZp\dto] \lto a[rXs_1\udto][s_1Zp\dto] \por \dots \por a[rXs_j\udto][s_jZp\dto] \por \dots\]
Note that the order of the $M$'s states $s_i$'s will not affect our construction.

We describe on Figure~\ref{Fig:DPPDAtoPEG} the rules for nonterminals depending on the rule $\delta(q, a, Z)$, where $q$ and $Z$ are fixed parameters and $a\in \Sigma$ runs the alphabet. For different $a$ we add different rules for the same nonterminal, the order of these rules is not significant.

\def\lefm{5pt}        
\begin{figure}
    \begin{center}        
    \setlength{\fboxsep}{2ex}
    \fbox{
    \begin{minipage}{0.8\textwidth}
    \setlength{\columnsep}{-0.5cm}
\begin{center}              
\textbf{Cases for $\delta(q, a, Z)$:}
\end{center}
\begin{multicols}{2}
\begin{enumerate}
    \item $(r, X, \rto)$: 
    \begin{itemize}[topsep=0pt, itemsep=1ex,leftmargin=\lefm]
        \item $[qZp\dto] \lto a[rXs\udto][sZp\dto]$ 
        \item $[qZ\bar p] \lto a[rXs\udto][sZ\bar p]$
        \item $[qZp\uto] \lto \&(a[rXs\udto][sZ\bar p])$
    \end{itemize}
    \columnbreak
    \item $(r, X, \dto)$:
    \begin{itemize}[topsep=0pt,itemsep=1ex,leftmargin=\lefm]
        \item $[qZp\dto] \lto (\&a)[rXs\udto][sZp\dto]$
        \item $[qZ\bar p] \lto (\&a)[rXs\udto][sZ\bar p]$
        \item $[qZp\uto] \lto (\&a)\&([rXs\udto][sZ\bar p])$
    \end{itemize}
    \end{enumerate}
\end{multicols}    
\begin{multicols}{3}
    \begin{enumerate}
        \setcounter{enumi}{2}
    \item $(p, \eps, \uto)$:
    \begin{itemize}[topsep=0pt, leftmargin=\lefm]
        \item $[qZp\dto] \lto \pfail$
        \item $[qZ\bar p] \lto \&a$
        \item $[qZp\uto] \lto \&a$
    \end{itemize}
    \columnbreak
    \item $(p, \eps, \dto)$:
    \begin{itemize}[topsep=0pt, leftmargin=\lefm]
        \item $[qZp\dto] \lto \&a$
        \item $[qZ\bar p] \lto \pfail$
        \item $[qZp\uto] \lto \pfail$
    \end{itemize}
    \columnbreak
    \item $(r, \eps, \udto)$; below $r\neq p$:
    \begin{itemize}[topsep=0pt, leftmargin=\lefm]
        \item $[qZp\dto] \lto \pfail$ 
        \item $[qZ\bar p] \lto \pfail$
        \item $[qZp\uto] \lto \pfail$
    \end{itemize}
    \end{enumerate}
\end{multicols}    
\end{minipage} }    \end{center}

\caption{Construction of the PEG by the $\delta$'s description}\label{Fig:DPPDAtoPEG}
\end{figure}

\begin{lemma}\label{lem:ArbitraryNorderPEGfromDPDA}
    On the first step of the computation of $R(A, au)$, where $A \in \{[qZp\udto], [qZ\bar p]\}$ the rules added for $\delta(q, b, Z)$, $b\neq a$, yields $\pfail$. So the order of rules for each non-terminal does not matter
\end{lemma}
\begin{proof}
    Each right side of a rule is either $\pfail$ or begins with $a$ or $\& a$ for the corresponding $a$ from $\delta(q, a, Z)$.
\end{proof}

The following lemma directly follows from the construction and Lemma~\ref{lem:ArbitraryNorderPEGfromDPDA}, so we omit the proof.


\begin{lemma}\label{lem:PropsForDPDAtoPEG}
     For each $u\in \Sigma^*$ 
     \begin{enumerate}
         \item Either $R([qZp\uto], u) = u$ or $R([qZp\uto], u) = \pfail$.
         \item $R([qZp\uto], u) = u \iff R([qZ\bar p], u) \neq \pfail$.
     \end{enumerate}
\end{lemma}

For technical needs, we need conditions of the form
\begin{equation}\label{eq::run:form}
  \forall i' : (q, Z \times i', i) \ders (p, (), d_j[i'])  
\end{equation}
where function $d_j[\cdot]$ is defined as follows: $d_j[i'] = j$ in the case $d_{q,i}(Z) =\,\dto$ and $d_j[i'] = i'$ in the case $d_{q,i}(Z) =\,\uto$. 
Note that the function $d_j[i']$ is either the constant or the identity function depending on $q, Z, i$ (and the input word); we write $d^{q,Z,i}_j[i']$ when it is needed to avoid ambiguity.

\begin{theorem}\label{thm:DPDAtoPEG}
    For each input $w = w_1\cdots w_n, w_i  \in \Sigma$ the following assertions hold.
    \begin{enumerate}
        \item  $\exists j \geq i : \!\big[R([qZp\dto], w_i\cdots w_n) = w_{j}\cdots w_n \iff \forall i' : \![(q, Z \times i', i) \ders (p, (), j)] \land d_{q,i}(Z) =\,\dto\!\big]$
        \item  $\exists j \geq i : \big[ R([qZ\bar p], w_i\cdots w_n) = w_j\cdots w_n \iff \forall i' : [(q, Z \times i', i) \ders (p, (), i')] \land d_{q,i}(Z) =\,\uto\big]$ 
        \item  $R([qZp\uto], w_i\cdots w_n) = w_{i}\cdots w_n \iff \forall i' : [(q, Z \times i', i) \ders (p, (), i')] \land d_{q,i}(Z) =\,\uto$
    \end{enumerate}
\end{theorem}

We assume that $j=n+1$ in the case $R([qZp\dto], w_i\cdots w_n) = \eps$, i.e., when the PEG has processed the whole (rest of) input, and vice versa (when the head is over the position~$n+1$, PEG processed the whole input).

\begin{proof}
    According to Lemma~\ref{lem:PropsForDPDAtoPEG} Assertion~3 follows from Assertion~2, so we prove only Assertions 1-2
that can be written as a single assertion
\begin{equation}
 \exists j \geq i : \big[R([qZ\hat p], w_i\cdots w_n) = w_{j}\cdots w_n \iff \forall i' : (q, Z \times i', i) \ders (p, (), d_j[i'])\big]    
\end{equation}
where $\hat p = p\dto$ if $d_{q,i}(Z) = \dto$ and $\hat p = \bar p$ if $d_{q,i}(Z) = \uto$.
    
    Proof for the implication ($\Rightarrow$) is by induction on PEGs derivation steps. The base case is Cases 3-5 and it is obvious, so we describe briefly only Case 3. 
For the rule $[qZp\dto] \lto \pfail$ the pop direction of the PEG rule does not coincide with the pop direction of the DPPDA move, so both sides of implication in Assertion~1 are false. For the rule $[qZ\bar p] \lto \&a$ in the Assertion~2 $r=i$ and  DPPDA does exactly one move $(q, Z \times i', i) \der (p, (), i')$, so both sides of the implication hold. Cases 4-5 are similar. 

Now we prove the induction step. In Cases 3-5 we have derivations of length 1 that are the base case, so in the case of the induction step only Cases 1-2 hold. Let $w_i = a$. In Case 1, if $R([qZ\hat p], w_i\cdots w_n) = w_j\cdots w_n$, we have that there exist positions $j, k$ and a state $s$ such that $R([rXs\udto], w_{i+1}\cdots w_n) = w_k\cdots w_n$ and $R([sZ\hat p], w_k\cdots w_n) = w_j\cdots w_n$. So by induction hypothesis, we have that 
\[ R([rXs\udto], w_{i+1}\cdots w_n) = w_k\cdots w_n \Rightarrow (r, X \times (i+1), i+1) \ders (s, (), k) \text{\quad and}\]
\[ R([sZ\hat p], w_{k}\cdots w_n) = w_j\cdots w_n \Rightarrow \forall i': (s, Z \times i', k) \ders (p, (), d_j[i']) \]
Combining all together, we obtain 
\[\forall i':  (q, Z \times i', i) \der (r, XZ \times (i+1, i'), i+1) \ders (s, Z \times i', k) \ders  (p, (), d_j[i'])\]
that proves Case 1. Case 2 differs from Case 1 only by the first move, so we omit the proof since it is the same.

Proof for the implication ($\Leftarrow$) is by induction on DPPDA moves. 
So, the assertion
\begin{equation}\label{ProofDPPDAtoPEG::indInvL}
  \forall i' : (q, Z \times i', i) \ders (p, (), d_j[i'])  
\end{equation}
holds.
 
The base case is the computation of a single move (Cases 3-5). Since each move is either a push or a pop, we have that this move is a pop. 

So if the pop direction $d_{q,i}(Z)$ is $\dto$ the (only) right side of Assertion 1 holds and it implies 
$R([qZp\dto], w_i\cdots w_n) = w_i\cdots w_n$ because of the rule $[qZp\dto] \lto \&a$.
If the pop direction is $\uto$ the right sides of Assertions 2-3 hold and the rules
$A \lto \&a$,  $ A \in \{[qZ\bar p], [qZp\uto]\}$ imply 
$R(A, w_i\cdots w_n) = w_i\cdots w_n$, so the base case hold.

Now we prove the induction step. If the first move is a pop, then we have the base case, so the first move is a push. So, the run~\eqref{ProofDPPDAtoPEG::indInvL} has a form 
\[ (q, Z \times i', i) \der (r, XZ \times (l, i'), l)  \ders (s, Z \times i' , m) \ders (p, (), d_j[i']) \]
where $l \in \{i, i+1\}$ depending on the direction $\dto$ or $\rto$ and $m$ is either a constant if $d_{r,l}(X) = \dto$ or $\forall l' : (r, X \times l', l)  \ders (s, (), l')$ if $d_{r,l}(X) = \uto$, so there exists a constant $k$ such that 
for $m = d^{r,X,l}_k[l']$  the relation 
\[\forall l' : (r, X \times l', l)  \ders (s, (), d^{r,X,l}_k[l'])\]
holds and for the same reason the relation
\[\forall j' : (s, Z \times j' , m) \ders (p, (), d^{s,Z,m}_{j}[j'])\]
holds as well. 
Induction hypothesis proves the implications, where 
\[\forall l' : (r, X \times l', l)  \ders (s, (), d_k[l']) \Rightarrow R([rXs\udto], w_l\cdots w_n) = w_{m}\cdots w_n \]
and 
\[\forall j' : (s, Z \times j' , m) \ders (p, (), d_j[j']) \Rightarrow R([sZ\hat p], w_m\cdots w_n) = w_j\cdots w_n \]
So $\forall i' : (q, Z \times i', i) \ders (p, (), d_k[i']) \Rightarrow $
\[ R([qZ\hat p],  w_i\cdots w_n) = R([sZ\hat p], R([rXs\udto], w_l\cdots w_n)) = w_j\cdots w_n\] and we have proved the induction step.
\end{proof}

\begin{corollary}
    $R(S, w) = \eps \iff (q_0, Z_0\times 1, 1) \ders (q_f, (), |w|+1)$
\end{corollary}
\begin{proof}
    $S \lto [q_0 Z_0 q_f^1\dto] \por \ldots \por [q_0 Z_0 q_f^m\dto]$ for all final states $q_f^i$. Due to Assertion~1 of Theorem~\ref{thm:DPDAtoPEG} at most one of nonterminal $[q_0 Z_0 q_f^1\dto]$ parses $w$ and it does it iff $M$ comes from the initial configuration to an accepting one.
\end{proof}

\section{Linear-Time Simulation of 2DPPDA}

Our linear-time simulation algorithm for 2DPPDA is almost the same as S.~Cook's algorithm for 2DPDA~\cite{Cook1970_2DPDA}. One can find the detailed exposition in~\cite{AHU1974design} and \cite{Gluck_2013}. We describe here the algorithm on the general level; firstly we provide the formal statement of the section's result. 

\begin{theorem}\label{theorem:CooksAnalogue}
    Let $M$ be a 2DPPDA. The language $L(M)$ is recognizable in time $O(n)$ in the RAM model.
Moreover, there exists an $O(|w|)$ (in RAM) simulation algorithm for $M$ on the input $w$.
\end{theorem}

We begin with definitions. A \emph{surface configuration} of and 2DPPDA is a triple $(q, A, i)$ of the current state $q$, the symbol on the top of the stack $A$, and the head's position $i$.
 Note that since the description of the automaton is fixed, the number of surface configurations is $O(n)$.
 Any configuration $(q, A\alpha\times \vec{j}, i)$ has the corresponding surface configuration $(q, A, i)$. We omit $\vec{j}$ from the configuration for the sake of notation: so, each surface configuration can be considered as a configuration (even if this configuration is unreachable on the processing of the input). Since the values of $\vec{j}$ are used only to determine the head's arrival position after $\uto$ moves, our ignorance does not affect the following definitions and constructions.  

So, we can run the automaton starting from a surface configuration $C = (q, A, i)$. We define relations $\mder$ and $\yield$ on surface configurations as follows. Let $D = (p, B, j)$. We say that $C \models D$ if $(q, A, i) \der (p, BA, j)$, i.e., the automaton pushes $B$ to the stack at the surface configuration $C$.
We say that $C \yield D$ if $(q, A, i) \ders (p, B, j)$ and there is no configuration $(p',B',j')$ such that $(q, A, i) \ders (p',B',j') \ders (p, B, j)$. So $C \yield D$ means that starting the computation in the surface configuration $C$ with the stack height $h$ the automaton firstly returns to the height $h$ in the surface configuration $D$. We denote by $\mders$ and $\yields$ reflexive and transitive closures of the relations $\models$ and $\yield$ respectively. A surface configuration $D = (p, B, j)$ is a \emph{terminator} for the configuration $C$ if $C \yields D$ and the automaton pops $B$ on the move right after $D$. Denote the terminator of the surface configuration $C$ as $T(C)$. Note that it is possible that $T(C) = C$.

The idea of the linear-time simulation algorithm is as follows. If we compute a terminator for each surface configuration (maybe for some surface configurations we find that they have no terminators), then we have computed the terminator $T_0 = (q_T, Z_0, i)$ for the initial surface configuration $C_0 = (q_0, Z_0, 0)$. If $q_T \in F$ the input is accepted, otherwise it is rejected. Terminators are computed for all surface configurations reachable from $C_0$ via the following dynamical programming algorithm. If during a recursive call, it was computed that $T(C) = D$, then the result stored in a memoization table $T[C] := D$, and if there would be another recursive call $T(C)$, the result will be returned in $O(1)$. We also use a memoization table to store the information, whether $T()$ was called for the configuration $C$. If at some point the algorithm shall compute $T(C)$ and $T(C)$ has been called previously, it means that the automaton has come to an infinite loop, so we terminate the algorithm with the rejection of the input.

So, the initial configuration for which the computations start is $C = C_0$. If at surface configuration $C = (q, A ,i)$ the action is pop, then $T(C) = C$. Otherwise $C \mder D$ for some surface configuration $D$ (computable in $O(1)$ via the transition table). So we compute $T(D) = (p, B , j)$ recursively, get the state $p'$ and the position $j'$ after the automatons move at $T(D)$ and obtain that $C \yield C' = (p', A , j')$, then we recursively compute $T(C')$ and get that $T(C) = T(C')$. So, we have finished the algorithm's description.

In fact, the only difference with S.~Cook's algorithm is that $j'$ can be equal to $i$ if the DPPDA returns the head to the cell of the $B$'s push. But this difference does not affect the construction. More precisely, assume that the configuration $C'$ is computed via the function $f$ that depends on configurations $C$ and $D$. But the exact arguments of $f$ vary between S.~Cook's and our algorithms: in the former case $C' = f(A, D)$, in the latter case $C' = f(A, i, D)$. Also for our construction important to deal with only surface configurations corresponding to the pushes. For a classical 2DPDA, it does not matter whether $A$ was pushed at $i$ if we begin the computation from $(q, A, i)$, but for 2DPPDA it is important.

Let us analyze the algorithm's complexity. Note that for each surface configuration $C$ we make at most $2$ recursive calls, so since there are $O(n)$ surface configurations, the total number of calls is $O(n)$. Since computations at each call take $O(1)$, we conclude that the whole algorithm works in $O(n)$.

\section{Structural Results}

We use the computational model to obtain new structural results about the PELs. 

\begin{lemma}\label{Lemma:LeftConcatDCFL}
    Let $X$ be a DCFL and $Y$ be a PEL. Then $XY$ is a PEL. 
\end{lemma}
\begin{proof}
    We describe a DPPDA $M$ recognizing $XY$ that simulates a DPDA $M_X$ recognizing $X$ and a DPPDA $M_Y$ recognizing $Y$, constructed by a (well-formed) PEG.
    
    DPPDA $M$ simulates $M_X$ until it reaches an accepting state. Then it pushes the information of the state to the stack, then pushes $Z_0$ (of $M_Y$) and simulates $M_Y$. If $M_Y$ accepts the rest of the input, then the whole input is accepted. Otherwise, $M$ pops symbols from the stack until it reaches the info about the $M_X$ state and continues the simulation until it reaches an accepting state again. This process is continued until either $M_Y$ accepts, or $M_X$ reaches the end of the input (and $M_Y$ rejects $\eps$).
    
    The correctness easily follows from the construction. During the process $M$ tests all the prefixes of the input from $X$ and checks whether the corresponding suffixes belong to $Y$.         
\end{proof}

We use the notation $\PEL$, $\DCFL$, and $\REG$ for the language classes in formulas (the last one denotes regular languages).
Denote by $\Gamma_{\REG}(\DCFL)$ the regular closure of DCFLs; this class is defined as follows. 

$L \in \Gamma_{\REG}(\DCFL)$ if there exists a regular expression (RE) $R$ over an alphabet $\Sigma_{m} = \{a_1, \ldots, a_m\}$ and DCFLs $L_{a_1}, \ldots, L_{a_m}$ such that if we replace $a_i$ by $L_{a_{i}}$ in $R$ the resulting expression $\psi(R)$ describes $L$.

\begin{lemma}\label{Lemma:GammaRegDCFL}
   $\Gamma_{\REG}(\DCFL) \subseteq \PEL$.
\end{lemma}

Firstly we describe the proof idea. 
We provided the proof of Lemma~\ref{Lemma:LeftConcatDCFL} to generalize it as follows. 
In the case of a single concatenation, we have a kind of linear order for an exhaustive search.
In the case of $\Gamma_{\REG}(\DCFL)$ we will perform an exhaustive search in the order corresponding to a 
(graph of) deterministic finite automaton (DFA) recognizing $R$. If the input word~$w$ belongs to $L \in \Gamma_{\REG}(\DCFL)$, 
then it can be split into subwords $w_1\cdots w_k = w$ such that there exists a word $\alpha_1 \cdots \alpha_k $ generated by $R$, $\alpha_i \in \Sigma_m$,  such that $w_i \in L_{\alpha_i}$.
So, the exhaustive search finds the split of $w$ by considering $\alpha_1 \cdots \alpha_k$ in the length-lexicographic order and considering $w$'s subwords  
$w_i \in L_{\alpha_i}$ ordered by the length. If a word $w_1 \in L_{\alpha_1}$ is the shortest prefix, the DPPDA tries to find the shortest $w_2 \in L_{\alpha_2}$ and so on.
If at some point the DPPDA failed to find $w_{j+1} \in L_{\alpha_{j+1}}$, it rollbacks to $\alpha_{j}$ and tries to find a longer word $w_j \in L_{\alpha_{j}}$. If it fails, then it rollbacks to $\alpha_{j-1}$ and so on. During the search of $w_j$, the DPPDA simulates a DPDA~$M_{\alpha_j}$ recognizing $L_{\alpha_{j}}$.

\begin{proof}
    Fix a language $L \in \Gamma_{\REG}(\DCFL)$.
    Let $R$ be the RE over $\Sigma_m$ and $L_{a_1}, \ldots, L_{a_m} \in \DCFL$ such that $\psi(R) = L$.
    Let $\A$ be a complete
    \footnote{The transition function is defined for each pair of a state and a letter.} DFA recognizes the language generated by~$R$.
    We consider the graph representation of $\A$.
    Assume that each language $L_i$ does not contain an empty word. 
    Otherwise we can add $\eps$ transitions to $\A$ that duplicate transitions labeled by $a_i$ (if $\eps \in L_i$)
    and convert the resulting NFA to a DFA (this transformation justifies the assumption in our construction below).
    
    Now we describe a DPPDA $M$ recognizing $L$. Let $M_{a_1}, \ldots, M_{a_m}$ be DPDAs recognizing $L_{a_1},\ldots, L_{a_m}$.
    Moreover, we assume that $M_{a_i}$ never goes to an infinite loop (it is well known that such automata exist, a construction of such automata is described in~\cite{ShallitSC}), 
    so after each move we know, whether $M_{a_i}$ accepts the current input's prefix or not.
    
    
    Initially, $M$ has $Z_0$ in the stack and pushes to the stack a pair $(q_0, a_1)$ that indicates that $M$ tries to go from $q_0$ by a word from the language $L_1$. 
    If $M$ has a pair $(q_i, a_j)$ on the top of the stack, then $M$ simulates $M_{a_j}$ (starting from the current symbol over the head). When $M_{a_j}$ occurs in an accepting state $q_f$, 
    $M$ pushes to the stack the pair $\langle a_j, q_f \rangle$ that indicates that at this point  
    the simulation of $M_{a_j}$ has been paused at the state $q_f$, 
    then $M$ pushes the pair $(q_{i,j}, a_1)$ where $q_i \xrightarrow{a_j} q_{i,j}$ in the graph of $\A$ and continues the simulation.

If $M$ reached the right end marker during a simulation of $M_{a_j}$ and $M_{a_j}$ is not in an accepting state, or it is, but $q_{i,j}$ is not an accepting state of $\A$, then $M$ \emph{rollbacks}, that means the following.

During a rollback, $M$ pops from the stack all symbols used for the simulation of $M_{a_j}$ until reaches $(q_i, a_j)$.
Then it pops $(q_i, a_j)$ with pop-direction $\uto$ and pushes $(q_i, a_{j+1})$ if $j < k$. If $j = k$, then 
$M$ pops $(q_i, a_k)$, pops $\langle a_{j'}, q_f \rangle$ (if it exists) with pop direction $\uto$ and resumes the simulation of $M_{a_{j'}}$ 
(from the state $q_f$).
If during a rollback $M$ reaches $Z_0$, then $M$ rejects the input.

So, $M$ accepts the input iff during a simulation of $M_{a_j}$, $M_{a_j}$ reaches an accepting state and $q_{i,j}$ is $\A$'s accepting state as well. This condition is equivalent to $w \in L$. During the simulation, $M$ will exhaustively try all the words $u \in \Sigma^*_m$ for which possible $w \in \psi(u)$. Since $\eps\not\in L_i$, the length of $u$ is bounded by the length of $w$, so the search will terminate at some point. Since the length of the word accepted by $M_{a_j}$ during the simulation grows only if the suffix of the input cannot be accepted, none of the words from $\psi(u)$ of length at most $|w|$ would be skipped during the simulation, so the search is exhaustive.
\end{proof}

Denote by $\Gamma_{\Bool}(\CL)$ the \emph{Boolean closure} of the language's class $\CL$, i.e., $\Gamma_{\Bool}(\CL)$ is a minimal class satisfying the conditions:
\begin{itemize}
    \item  $\CL \subseteq \Gamma_{\Bool}(\CL)$ 
    \item  $\forall A, B \in \Gamma_{\Bool}(\CL): A \cup B, A\cap B, \overline{A} \in \Gamma_{\Bool}(\CL)$ 
\end{itemize}

\begin{theorem}\label{thm:main:struct} The following assertions hold.
        \begin{enumerate}
            \item $\Gamma_{\Bool}(\Gamma_{\REG}(\DCFL)) \subseteq \PEL$.
            \item  $\Gamma_{\REG}(\DCFL)\cdot\PEL = \PEL$.
        \end{enumerate}
\end{theorem}
\begin{proof}
    It was shown in~\cite{FordPEG2004} that $\Gamma_{\Bool}(\PEL) = \PEL$. We proved that $\Gamma_{\REG}(\DCFL) \subseteq \PEL$, so $\Gamma_{\Bool}(\Gamma_{\REG}(\DCFL)) \subseteq \PEL$.
    
    The inclusion $\Gamma_{\REG}(\DCFL)\cdot\PEL \supseteq \PEL$ is obvious ($\{\eps\} \in \Gamma_{\REG}(\DCFL)$).
    The inclusion $\Gamma_{\REG}(\DCFL)\cdot\PEL \subseteq \PEL$ follows from the modification of the simulation algorithm from the proof of Lemma~\ref{Lemma:GammaRegDCFL} by the simulation step from the proof of Lemma~\ref{Lemma:LeftConcatDCFL}: when $M_{a_j}$ reaches an accepting state and the state $q_{i,j}$ is an accepting state of $\A$, $M$ simulates DPPDA for the PEG. If it successfully parses the suffix, the input is accepted, otherwise, the simulation continues as in the proof of Lemma~\ref{Lemma:GammaRegDCFL}.
\end{proof}

\begin{corollary}
    For each $ L \in \Gamma_{\Bool}(\Gamma_{\REG}(\DCFL))$ there exists a RAM-machine $M$ that decides, whether $w \in L$ in $O(|w|)$. In other words, the class $\Gamma_{\Bool}(\Gamma_{\REG}(\DCFL))$ is linear-time recognizable.
\end{corollary}
\begin{proof}
 By Theorem~\ref{thm:main:struct}, $\Gamma_{\Bool}(\Gamma_{\REG}(\DCFL)) \subseteq \PEL$.
 There are several linear-time simulation algorithms known for PELs~\cite{TDPL_BirmanUlman, FordPEG2002}. Our automata-based construction with the presented linear-time simulation algorithm for 2DPPDA provides a constructive proof of this corollary. 
\end{proof}

\bibliography{PEG}

\end{document}